\documentclass[12pt,psfig,a4]{article}

\makeatletter
\def\@seccntformat#1{\@ifundefined{#1@cntformat}%
   {\csname the#1\endcsname\quad}  
   {\csname #1@cntformat\endcsname}
}
\let\oldappendix\appendix 
\renewcommand\appendix{%
    \oldappendix
    \newcommand{\section@cntformat}{\appendixname~\thesection:\,\,}
}
\makeatother

\usepackage{amsfonts}
\usepackage{amsmath}
\usepackage{amssymb}
\usepackage{amsthm}
\usepackage[toc,page]{appendix}
\usepackage{booktabs}
\usepackage{color,soul}
\usepackage{colortbl}
\usepackage{enumitem}
\usepackage{geometry}
\usepackage{graphics}
\usepackage{graphicx}
\usepackage{hhline}
\usepackage{hyperref}
\usepackage{enumitem}
\usepackage{epstopdf}
\usepackage{lscape}
\usepackage{multirow}
\usepackage[round, sort&compress]{natbib}
\usepackage{setspace}
\usepackage{subfigure}
\usepackage{tcolorbox}
\usepackage{xcolor}

\PassOptionsToPackage{
        natbib=true,
        style=authoryear-comp,
        hyperref=true,
        backend=biber,
        maxbibnames=99,
        firstinits=true,
        uniquename=init,
        maxcitenames=2,
        parentracker=true,
        url=false,
        doi=false,
        isbn=false,
        eprint=false,
        backref=true,
            }   {biblatex}

\newcommand{\Lyx}{L\kern-.1667em\lower.25em\hbox{y}\kern-.125emX\spacefactor1000}
\newtheorem{theorem}{Theorem}

\newtheorem{proposition}[theorem]{Proposition}

\newcommand{\E}{\mathbb{E}}

\newcommand{\V}{\mathbb{V}\mathrm{ar}}

\newcommand{\erf}{\mathrm{erf}}
\newcommand{\erfc}{\mathrm{erfc}}

\DeclareMathOperator{\diag}{diag}

\newcommand{\norm}[1]{\ensuremath{ \left\Vert #1 \right\Vert }}

\newcommand{\overbar}[1]{\mkern 1.5mu\overline{\mkern-1.5mu#1\mkern-1.5mu}\mkern 1.5mu}

\def\prehp(#1,#2){\ensuremath{  #1 \cdot #2 }}

\newcommand{\Cag}[1]{\textcolor{blue}{#1}}

\begin{document}

\title{
\textbf{Non-parametric cumulants approach for outlier detection of multivariate financial data}}
\author{Francesco Cesarone$^1$, Rosella Giacometti$^2$, Jacopo Maria Ricci$^3$\\
	{\small $^1$\emph{Roma Tre University  - Department of Business Studies}}\\
	{\footnotesize francesco.cesarone@uniroma3.it}\\
	{\small $^2$ \emph{University of Bergamo - Department of Management}}\\
	{\footnotesize rosella.giacometti@unibg.it} \\
	{\small $^2$ \emph{University of Bergamo - Department of Economics}}\\
	{\footnotesize jacopomaria.ricci@unibg.it} \\
}
\maketitle

\begin{abstract}

In this paper, we propose an outlier detection algorithm for multivariate data based on their projections on the directions that maximize the Cumulant Generating Function (CGF).
We prove that CGF is a convex function, and we characterize the CGF maximization problem on the unit $n$-circle as a concave minimization problem.
Then, we show that the CGF maximization approach can be interpreted as an extension of the standard principal component technique.
Therefore, for validation and testing, we provide a thorough comparison of our methodology with two other projection-based approaches both on artificial and real-world financial data.
Finally, we apply our method as an early detector for financial crises.

\end{abstract}

\section{Introduction}

Outlier detection issue has become increasingly important over the years and, as of now, its fields of application range from medicine and engineering to finance \citep[see, e.g.,][]{singh2012outlier,meng2019overview}.
As for the latter, outliers can be the consequence of human error or fraudulent activities; similarly, financial crises can be viewed as anomalies since markets experience atypical behaviors in those periods \citep[see, e.g.,][]{ane2008robust}.
Furthermore, a few outliers can strongly influence the results of an experiment. This can be observed in Portfolio Optimization, where some portfolio selection models can be highly sensitive to changes in input data \citep[see, e.g.,][]{Kondor2007,cesarone2020stability, giacometti2021tail}.
Because of this widespread practical relevance, many authors tackled this topic. Hence, the theory behind anomaly detection has unsurprisingly evolved, from the first studies which dealt with more simple instances, i.e., univariate Gaussian data, to more complex cases, such as multivariate data following non-parametric distributions.
We can mention the work of \cite{ferguson1961rejection}, who considers univariate normal samples and identifies outliers as data with mean slippage. \cite{wilks1963multivariate} proposes a method to identify outliers in a multivariate normal distribution with unknown parameters.
\cite{gnanadesikan1972robust} propose outlier detection methods on multivariate data, based on their projection onto the directions corresponding to the principal components obtained by the standard Principal Component Analysis (PCA).
\cite{schwager1982detection} extend the work of \cite{ferguson1961rejection} to multivariate normal data. \cite{reed1990adaptive} identify outliers by using the so-called RX detector.
Such an indicator measures the location of multivariate data points in the dispersion ellipsoid by means of the Mahalanobis distance. The authors assume that both ordinary and outlier data follow a multivariate normal distribution with same covariance matrix but different mean vectors.
Note that, even though \cite{reed1990adaptive} assume the Gaussian distribution, other authors used the RX detector, removing such a strong distributional hypothesis.
For example, \cite{das1986detection} and \cite{sinha1984detection} extend the work by \cite{schwager1982detection} to (nonnormal) elliptically symmetric distributions, and use the Mahalanobis distance as a detector.
Since outliers affect both location and scale parameters, one way to overcome this issue is to use robust estimates.
For example, \cite{maronna1976robust} proposes robust M-estimators of multivariate mean vectors and covariance matrices. Both \cite{stahel1981breakdown} and \cite{donoho1982breakdown} define, in their respective works, an affine equivariant robust estimator for location and scatter.
Another popular method consists in computing the minimal covariance determinant \citep[see, e.g.,][]{rousseeuw1984least, rousseeuw1985multivariate}.
%

%
%
%
%
%

\noindent
Especially when dealing with high-dimensional multivariate data, many techniques aim to find outliers in the univariate projections of such data to reduce the computational effort.
For this reason, several studies have been devoted to identifying the ``best'' directions in which data must be projected to represent the variability of data most effectively.
In \cite{pena2001multivariate, pena2007combining}, these directions correspond to those that maximize and minimize the kurtosis of the projected data. \cite{domino2020multivariate} generalizes the Prieto and Pe\~na's approach by
choosing the direction that maximizes the fourth cumulant.
Following this stream of literature, our work aims at detecting outliers,  by projecting the data onto the direction that maximizes the Cumulant Generating Function (CGF).
%

\noindent
In this paper, we refine some theoretical results of the methodologies proposed by \cite{bernacchia2008detecting1} and \cite{bernacchia2008detecting2}. More precisely, we prove that CGF is a convex function, and, then, we characterize the CGF maximization problem on the unit $n$-circle as a concave minimization problem.
Then, we extend the outlier detection methodology, based on the projections of multivariate data on the directions obtained by the classical PCA technique, to the directions that maximize CGF.
Finally, we perform an extensive empirical analysis both on simulated and historical data, comparing our method with those described in \cite{pena2001multivariate, pena2007combining} and \cite{domino2020multivariate}.

The paper is organized as follows.
In Section \ref{sec:cumulantsoptdir}, we first introduce some preliminary concepts about the moment generating function, the cumulant generating function, and cumulants for univariate and multivariate random variables.
Then, we report and refine some results developed in \cite{bernacchia2008detecting1} and \cite{bernacchia2008detecting2} for a generalization of the principal component analysis (PCA) technique.
In Section \ref{sec:outdet}, we describe our outlier detection algorithm, and the methods used to comparative purposes.
For validation and testing, in Section \ref{sec:expana}, we present a thorough comparison of these methodologies
both on artificial and real-world data.
The real-world data application identifies the outliers as the materialization of financial crises.
Finally, Section \ref{sec:conc} contains some concluding remarks.

\section{Theoretical framework}\label{sec:cumulantsoptdir}

For the sake of completeness and readability, we recall below some notions about the moment generating function, the cumulant generating function, and the cumulants in the case both of a univariate and a multivariate random variable. Furthermore, we also report and refine some concepts introduced by \cite{bernacchia2008detecting1} and \cite{bernacchia2008detecting2} for a generalization of the principal component analysis (PCA) technique.

\noindent
Let $X$ be a univariate random variable.
If $\E[e^{\xi X}]$ exists and is finite $\forall \xi \in \mathbb{R}$, then the moment generating function of $X$ is defined as follows
\begin{equation}
	\label{uniMGF}
		M_{X}(\xi)=\E[e^{\xi X}]=\sum_{m=0}^{+\infty}\mu_{m}\frac{\xi^{m}}{m!}
\end{equation}
where $\mu_{m}=\E[X^{m}]$ is the $m^{th}$ raw moment of $X$, that can be obtained by differentiating $m$ times w.r.t. $\xi$ and setting $\xi=0$,
namely $\mu_{m}=M_{X}^{(m)}(0)$.
The cumulant generating function (CGF) of $X$ can be expressed as the logarithm of \eqref{uniMGF}
\begin{equation}
G_{X}(\xi)=\ln \E[e^{\xi X}]=\sum_{m=1}^{+\infty}k_{m}\frac{\xi^{m}}{m!},
\label{CGFUni}
\end{equation}
where $k_{m}$ is the $m^{th}$ order cumulant of $X$.
It is straightforward to see that $k_{m}$ can be obtained by differentiating Expression \eqref{CGFUni} $m$ times and setting $\xi=0$, namely
$k_{m}=G_{X}^{(m)}(0)$.
Furthermore, the cumulants can be expressed as functions of the moments.
For instance, the first 4 cumulants of $X$ are as follows
\begin{equation*}
\begin{array}{l}
k_{1}=\E[X]\\
k_{2}=\E[(X-\E[X])^{2}]\\
k_{3}=\E[(X-\E[X])^{3}]\\
k_{4}=\E[(X-\E[X])^{4}]-3(\E[(X-\E[X])^{2}])^{2}\\
\end{array}
\end{equation*}
Note that the first two cumulants correspond to the mean and variance respectively, the third coincides with the third central moment, while from the fourth onward the cumulants are polynomial functions of the central moments with integer coefficients.
From Expression \eqref{CGFUni}, we also observe that for small values of $\xi$ the cumulant generating function of $X$ is essentially determined by its variance (if $k_{1}=0$), whereas for large $\xi$ the contributions of the higher order cumulants become dominant.
Furthermore, if $X\sim N(\mu, \sigma^{2})$, then
$G_{X}(\xi)= \mu \xi + \sigma^{2} \displaystyle \frac{\xi^{2}}{2}$ since $k_{m}=0$ for $m>2$.
	
\noindent
In the case of a multivariate random variable $\boldsymbol{X}=(X_{1},\ldots, X_{n})$, its CGF becomes
	\begin{equation}
	G_{\boldsymbol{X}}(\boldsymbol{\xi})=\ln \E[e^{\boldsymbol{\xi}^{T} \boldsymbol{X}}],
	\end{equation}
where $\boldsymbol{\xi}=(\xi_1, \ldots, \xi_n) \in \mathbb{R}^{n}$.
Denoting by
$r=\| \boldsymbol{\xi}\|_{2}$ the euclidian norm of $\boldsymbol{\xi}$ and by $\boldsymbol{\theta}$
the versor of $\boldsymbol{\xi}$, we can write
$\boldsymbol{\xi}=r\boldsymbol{\theta}$, and therefore
the cumulant generating function of $\boldsymbol{X}$
can be expressed as follows
\begin{equation}
\label{multiCGF_theta}
G_{\boldsymbol{X}}(r,\boldsymbol{\theta})=\ln \E[e^{r \boldsymbol{\theta}^{T} \boldsymbol{X}}]=
\sum_{m=1}^{+\infty} k_{m}(\boldsymbol{\theta})\frac{r^{m}}{m!}.
\end{equation}
where
$k_{m}(\boldsymbol{\theta})  = \displaystyle \frac{d^{m} G_{\boldsymbol{X}}(r,\boldsymbol{\theta})}{ d r^{m}} \biggl \lvert_{r = 0}$ is the $m^{th}$ cumulant of $\boldsymbol{X}$ projected along the direction $\boldsymbol{\theta}$.

\noindent
Note that the length $r$ of $\boldsymbol{\xi}$ plays the same rule of $\xi$ in the univariate case, namely
if $r$ is small then the information on $\boldsymbol{X}$ through CGF is basically represented by the covariance matrix of $\boldsymbol{X}$ (if data are centered, i.e., $\E[X_j]=0$ with $j=1, \ldots, n$).
Whereas if $r$ is large then the information about $\boldsymbol{X}$ described by its CGF predominantly depends on the higher-order cumulants.

\subsection{PCA through CGF}\label{sec:PCA_CGF}

Following the work of \cite{bernacchia2008detecting1}, we look for identifying the largest principal components, namely the directions that provide most of the variability present in the original data, by exploiting the information contained in the Cumulant Generating Function (CGF).

\noindent
Using the classical PCA technique the principal component with the largest eigenvalue is the versor maximizing the variance, which, as we will show in the next section, also coincides with the optimal versor maximizing CGF \eqref{multiCGF_theta} in the case of multivariate normal random variables, or in the case of generic random vectors but for small values of $r$.
As mentioned by \cite{bernacchia2008detecting1}, the CGF maximization aims to find the directions with the largest variability of the multivariate distribution not only through the first two cumulants, but also through the higher-order ones, which become dominant especially for data points that deviate widely from the mean of the random vector.
Hence,
for a fixed $r$, our goal is to find the directions that maximize \eqref{multiCGF_theta}, i.e.
\begin{equation}\label{MaxCGF}
\begin{array}{lll}
\displaystyle\max_{\boldsymbol{\theta}} & G_{\boldsymbol{X}}(r,\boldsymbol{\theta}) \\
\mbox{s.t.}& &\\
& \boldsymbol{\theta}^{T}\boldsymbol{\theta}=1 &
\end{array}
\end{equation}
As we will show in Section \ref{theta_gen}, Problem \eqref{MaxCGF} is a constrained concave programming problem,
where the main difficulty lies in the fact that concave problems normally have many local maxima points.
\noindent
In the next section, we present our approach on two artificial multivariate random vectors, normal and skew-normal.

\subsection{Some special cases: multivariate normal and skew-normal}

To better understand the rationale behind the approach of \cite{bernacchia2008detecting1}, we report the CGF maximization procedure in the case of multivariate normal (Section \ref{sec:Normal}) and skew-normal (Section \ref{sec:Skew-Normal}) random variables, while in Section \ref{theta_gen} we address the general case.
We show that, in the case of a multivariate normal random vector, the optimal versor maximizing CGF collapses to the first principal component of the PCA.
For a skew-normal, we will see that the optimal solution obtained from Problem \eqref{MaxCGF} tends to the first principal component when the radius of the hypersphere is small, while, in the case of a large radius, the direction of the optimal versor deviates from the first component toward the direction where the multivariate distribution presents a fatter tail.

\noindent
Without loss of generality, we henceforth assume that the data are centered around their mean.


\subsubsection{Multivariate normal case}\label{sec:Normal}
		
Let $\boldsymbol{X} \sim N(\boldsymbol{0}, \boldsymbol{\Sigma})$, where, therefore, $\boldsymbol{X}$ is mean-centered and $\boldsymbol{\Sigma}$ is its covariance matrix.
Thus, the cumulant generating function of $\boldsymbol{X}$ is
	\begin{equation*}
	G_{\boldsymbol{X}}(\xi)=\frac{1}{2}\boldsymbol{\xi}^{T}\boldsymbol{\Sigma} \boldsymbol{\xi} \,.
	\end{equation*}
Indeed, similar to the univariate case, for normal random vectors, the first two cumulants are exactly the mean and variance,
while the higher-order cumulants are all equal to zero.
Then, Problem \eqref{MaxCGF} can be rewritten as
	\begin{equation}\label{MaxCGFNorm}
	\begin{array}{lll}
	\displaystyle\max_{\boldsymbol{\theta}} & \displaystyle\frac{r^{2}}{2}\boldsymbol{\theta}^{T}\boldsymbol{\Sigma} \boldsymbol{\theta}& \\
	\mbox{s.t.}& &\\
	& \boldsymbol{\theta}^{T}\boldsymbol{\theta}=1 &
	\end{array}
	\end{equation}
This problem can be easily solved by means of the Lagrange multipliers method, where the Lagrangian function is
$L(\lambda, \boldsymbol{\theta})=\displaystyle \frac{r^{2}}{2}\boldsymbol{\theta}^{T}\boldsymbol{\Sigma} \boldsymbol{\theta}-\lambda(\boldsymbol{\theta}^{T}\boldsymbol{\theta}-1)$.
Hence, the corresponding first order conditions are
$\nabla_{\lambda} L(\lambda, \boldsymbol{\theta}) = \boldsymbol{\theta}^{T} \boldsymbol{\theta} - 1 = 0$ and
$\nabla_{\boldsymbol{\theta}} L(\lambda, \boldsymbol{\theta})=r^{2}\boldsymbol{\Sigma}\boldsymbol{\theta}-2\lambda\boldsymbol{\theta}=0$.
Problem \eqref{MaxCGFNorm} has two optimal solutions, the versors $\boldsymbol{\hat{\theta}_1}$ and $\boldsymbol{\hat{\theta}_2}=-\boldsymbol{\hat{\theta}_1}$, which are both eigenvectors of $\boldsymbol{\Sigma}$.
Note that applying the classical PCA technique, the first principal components can be obtained by solving the following problem
		\begin{equation}\label{MaxSigma}
		\begin{array}{lll}
		\displaystyle\max_{\boldsymbol{\theta}} & \boldsymbol{\theta}^{T}\boldsymbol{\Sigma}\boldsymbol{\theta} \\
		\mbox{s.t.}& &\\
		& \boldsymbol{\theta}^{T}\boldsymbol{\theta}=1 &
		\end{array}
		\end{equation}
Since Problems \eqref{MaxCGFNorm} and \eqref{MaxSigma} are equivalent except for the parameter $\displaystyle\frac{r^{2}}{2}$ in the objective function of \eqref{MaxCGFNorm}, these problems lead to the same optimal solutions.
\begin{figure}[htbp]
\centering
\includegraphics[width=0.7\textwidth]{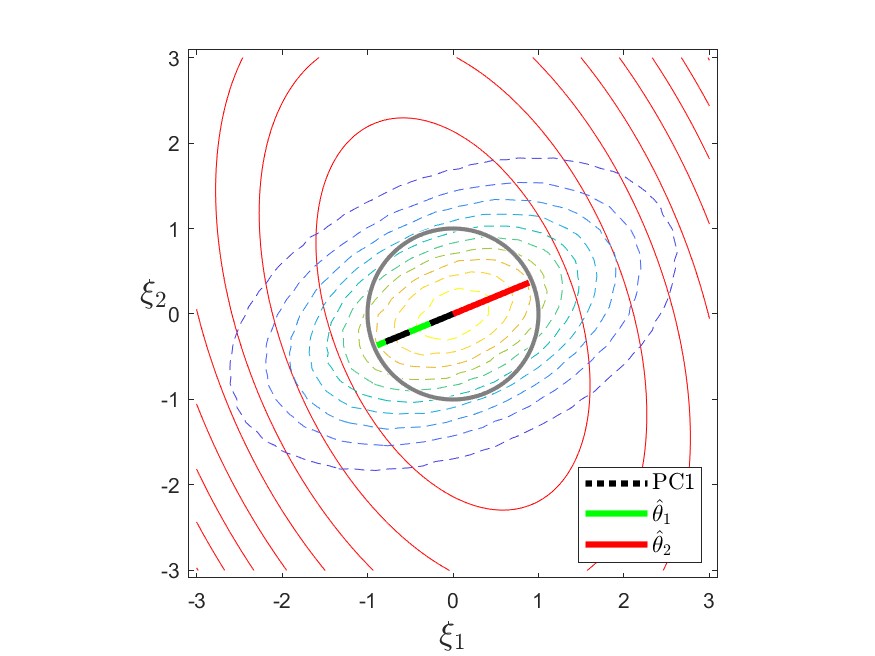}
\caption{Bivariate normal random variable: the dashed lines represent its isoprobability contours of $\boldsymbol{X}$, the red solid lines are the isocurves of
$G_{\boldsymbol{X}}(r, \boldsymbol{\theta}) = \frac{1}{2} \boldsymbol{\xi}'\boldsymbol{\Sigma} \boldsymbol{\xi}$,
and the gray solid line is the circle of radius 1.}
\label{fig:isopbinorm}
\end{figure}

\noindent
In Fig. \ref{fig:isopbinorm}, we give a graphical representation of the CGF maximization procedure in the case of a bivariate normal random variable.
The red solid lines are the isocurves of $G_{\boldsymbol{X}}(\boldsymbol{\xi})=\frac{1}{2} \boldsymbol{\xi}^{T}\boldsymbol{\Sigma} \boldsymbol{\xi}$,
the gray solid line is the ball of radius 1, and the dashed lines represent the isoprobability contours of $\boldsymbol{X} \sim N(\boldsymbol{0}, \boldsymbol{\Sigma})$.
The points, where the (red solid) isocurves and the (gray solid) cicle of radius 1 are tangent, represent the optimal solutions of Problem \eqref{MaxCGFNorm}, which coincide (except for a minus sign) with the first principal component lying on the major axis of the (dashed) local-dispersion ellipsoids, i.e., the direction of maximum variability of the bivariate random variable.
We point out that the fact that $\boldsymbol{\hat{\theta}_1}$ and $\boldsymbol{\hat{\theta}_2}$ lie on the same direction
is due to the symmetry of the bivariate normal distribution.
Indeed, as we will see in the next section, in the case of a non-symmetric distribution $\boldsymbol{\hat{\theta}_1}$ and $\boldsymbol{\hat{\theta}_2}$ are in general on different directions, and these directions depend on $r$.
Clearly, this phenomenon will be more evident for large $r$ and for highly skewed distributions.

\subsubsection{Multivariate skew-normal case}\label{sec:Skew-Normal}

We report here the implementation of the CGF maximization procedure in the case of a multivariate skew-normal random variable.

\noindent
Let $\boldsymbol{X} \sim SN(\boldsymbol{\eta}, \boldsymbol{\Sigma}, \boldsymbol{\alpha})$, where $\boldsymbol{\eta}$, $\boldsymbol{\Sigma}$ and $\boldsymbol{\alpha}$ are the location, the scale and the shape parameters, respectively.
As discussed in \cite{azzalini1996multivariate}, the probability density function (pdf) of $\boldsymbol{X}$ is
\begin{equation}\label{eq:pdfSN}
f_{\boldsymbol{X}}(\boldsymbol{x})=2 \phi_{n}(\boldsymbol{x}-\boldsymbol{\eta}; \boldsymbol{\Sigma})
\Phi(\boldsymbol{\alpha}^{T}\diag(\boldsymbol{\sigma})^{-1}\bigl(\boldsymbol{x}-\boldsymbol{\eta}\bigr)) \quad \mbox{with}
\quad \boldsymbol{x} \in \mathbb{R}^{n},
\end{equation}
where $\phi_{n}(\boldsymbol{x}-\boldsymbol{\eta}; \boldsymbol{\Sigma})$ is the pdf of an $n$-variate Gaussian random variable with mean $\boldsymbol{\eta}$ and covariance matrix $\boldsymbol{\Sigma}$, $\Phi(\cdot)$ is the cumulative distribution function (cdf) of a univariate standard normal random variable, $\boldsymbol{\alpha}$ is the degree of skewness (when $\boldsymbol{\alpha}=\boldsymbol{0}$, \eqref{eq:pdfSN} collapses to the normal case), and $\diag(\boldsymbol{\sigma})=\diag(\sigma_{1},\ldots,\sigma_{n})$, i.e., it represents the diagonal matrix of the standard deviations.
As shown in \cite{arellano2008centred} and \cite{azzalini1999statistical}, we have
\begin{equation}\label{eq:ExpSN}
  \boldsymbol{\mu_X} = \mathbb{E}[\boldsymbol{X}] = \boldsymbol{\eta} + \diag(\boldsymbol{\sigma}) \boldsymbol{\delta}
\end{equation}
\begin{equation}\label{eq:VarSN}
  \boldsymbol{\Sigma_X} = \mathbb{V}\mathrm{ar}[\boldsymbol{X}] = \boldsymbol{\Sigma} - \diag(\boldsymbol{\sigma}) \boldsymbol{\delta} \boldsymbol{\delta}^{T} \diag(\boldsymbol{\sigma})
\end{equation}
\begin{equation}\label{eq:CGF_SN}
  G_{\boldsymbol{X}}(\boldsymbol{\xi})=\ln \E[e^{\boldsymbol{\xi}^{T} \boldsymbol{X}}]  = \boldsymbol{\xi}^{T} \boldsymbol{\eta} + \frac{1}{2} \boldsymbol{\xi}^{T} \boldsymbol{\Sigma} \boldsymbol{\xi} + \ln [2 \Phi(\sqrt{\frac{\pi}{2}} \boldsymbol{\delta}^{T} \diag(\boldsymbol{\sigma})  \boldsymbol{\xi}] \, ,
\end{equation}
where $\boldsymbol{\delta} = \displaystyle\frac{1}{\sqrt{ \frac{\pi}{2} (1 + \boldsymbol{\alpha}^{T} \boldsymbol{C} \boldsymbol{\alpha})}} \boldsymbol{C} \boldsymbol{\alpha}$, and $\boldsymbol{C} = \diag(\boldsymbol{\sigma})^{-1} \boldsymbol{\Sigma} \diag(\boldsymbol{\sigma})^{-1}$.

\noindent
Since in this framework we work with a centered variable around its mean $\overbar{\boldsymbol{X}} = \boldsymbol{X} - \boldsymbol{\mu_X}$,
its CGF becomes
\begin{equation}
G_{\overbar{\boldsymbol{X}}}(\boldsymbol{\xi})= - \boldsymbol{\xi}^{T} \boldsymbol{\mu_X} + \boldsymbol{\xi}^{T} \boldsymbol{\eta} + \frac{1}{2} \boldsymbol{\xi}^{T} \boldsymbol{\Sigma} \boldsymbol{\xi} + \ln [2 \Phi(\sqrt{\frac{\pi}{2}} \boldsymbol{\delta}^{T} \diag(\boldsymbol{\sigma})  \boldsymbol{\xi}] \, .
\label{G_s_SN_a}
\end{equation}
Exploiting Expression \eqref{eq:ExpSN} and denoting $\boldsymbol{\hat{\mu}} = \boldsymbol{\mu_X} - \boldsymbol{\eta}$ and $\boldsymbol{\xi}=r\boldsymbol{\theta}$,
we can write
\begin{equation}
G_{\overbar{\boldsymbol{X}}}(r, \boldsymbol{\theta}) = - r \boldsymbol{\theta}^{T} \boldsymbol{\hat{\mu}} + \frac{r^2}{2} \boldsymbol{\theta}^{T} \boldsymbol{\Sigma} \boldsymbol{\theta} + \ln [ 2 \Phi(\sqrt{\frac{\pi}{2}} r \boldsymbol{\hat{\mu}}^{T}  \boldsymbol{\theta})],
\label{G_s_SN}
\end{equation}
For small $r$, using the Taylor expansion of $\ln (1 + z) = z - \frac{z^2}{2} + \ldots$ and
$2 \Phi(z) = 1+ \erf(\frac{z}{\sqrt{2}})= 1 + \sqrt{\frac{2}{\pi}} z + \ldots$,
we have $\ln [ 2 \Phi(\sqrt{\frac{\pi}{2}} r \boldsymbol{\hat{\mu}}^{T}  \boldsymbol{\theta})] = \ln (1+ \erf(\sqrt{\frac{\pi}{2}} r \boldsymbol{\hat{\mu}}^{T}  \boldsymbol{\theta})) \simeq \ln (1 + \sqrt{\frac{2}{\pi}} \sqrt{\frac{\pi}{2}} r \boldsymbol{\hat{\mu}}^{T}  \boldsymbol{\theta}) \simeq r \boldsymbol{\hat{\mu}}^{T}  \boldsymbol{\theta} - \frac{1}{2} r^2 \boldsymbol{\theta}^{T} \boldsymbol{\hat{\mu}} \boldsymbol{\hat{\mu}}^{T}  \boldsymbol{\theta}$.
Therefore, Expression \eqref{G_s_SN} can be approximated as follows
\begin{equation}
G_{\overbar{\boldsymbol{X}}}(r, \boldsymbol{\theta}) \simeq \frac{r^2}{2} \boldsymbol{\theta}^{T} \left( \boldsymbol{\Sigma} - \boldsymbol{\hat{\mu}} \boldsymbol{\hat{\mu}}^{T} \right) \boldsymbol{\theta},
\end{equation}
where $\boldsymbol{\Sigma} - \boldsymbol{\hat{\mu}} \boldsymbol{\hat{\mu}}^{T}$ is exactly the covariance matrix of the skew-normal distribution $\boldsymbol{\Sigma_X}$
as in \eqref{eq:VarSN}.
Hence, for small $r$, the optimal solution obtained from Problem \eqref{MaxCGF} collapses to the first principal component,
namely $\boldsymbol{\hat{\theta}_1}=PC1$ and $\boldsymbol{\hat{\theta}_2}=-PC1$.
Fig. \ref{fig:isopbiskewnormsmallr} exhibits a numerical example of the CGF maximization procedure in the case of a bivariate skew-normal random variable,
when $r$ is small.
	\begin{figure}[htbp]
\centering
		\includegraphics[width=0.7\textwidth]{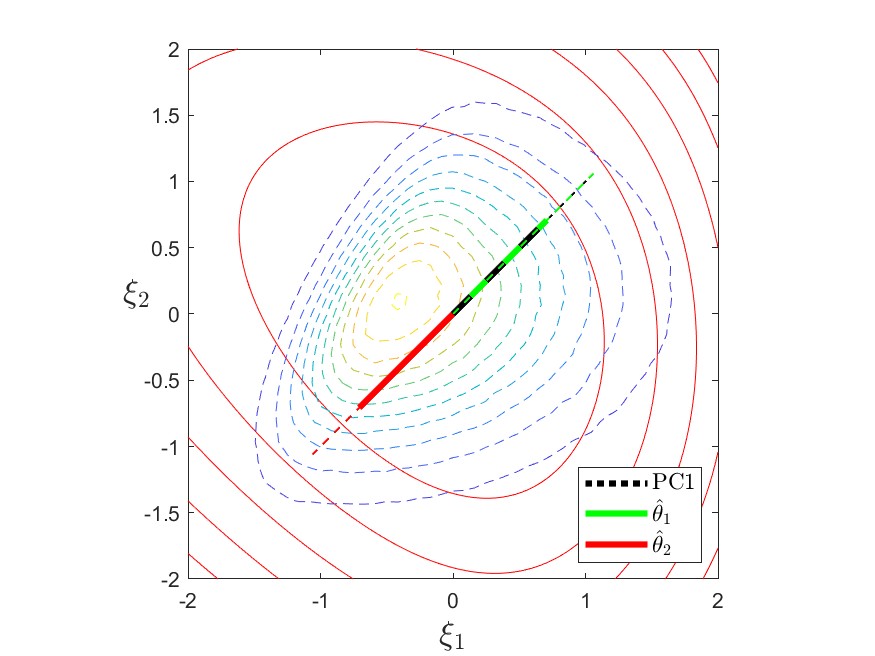}
		\caption{Bivariate skew-normal random variable: the dashed lines represent its isoprobability contours, the red solid lines are the isocurves of
			$G_{\boldsymbol{X}}(r, \boldsymbol{\theta})=-r\boldsymbol{\mu}^{T}\boldsymbol{\theta}+\frac{r}{2}\boldsymbol{\theta}^{T}\boldsymbol{\Sigma}\boldsymbol{\theta}+\ln [2\Phi(\sqrt{\frac{\pi}{2}}r\boldsymbol{\mu}^{T}\boldsymbol{\theta})]$, for small $r$.}
		\label{fig:isopbiskewnormsmallr}
	\end{figure}

\noindent
For large $r$, we can examine the behavior of Eq. \eqref{G_s_SN} using the asymptotic expansion of the error function \citep[see, e.g.,][]{copson2004asymptotic} and distinguishing the case where $\boldsymbol{\hat{\mu}}^{T}  \boldsymbol{\theta} \geq 0$ and $\boldsymbol{\hat{\mu}}^{T}  \boldsymbol{\theta}<0$.
For this aim, we consider the following asymptotic expansion of the complementary error function for large real $z$
\begin{equation}\label{eq:AsymErr}
 \erfc(z)= \frac{e^{-z^{2}}}{z {\sqrt {\pi }}} \sum _{n=0}^{\infty }(-1)^{n}{\frac {(2n-1)!!}{\left(2 z^{2}\right)^{n}}} =
 \frac{e^{-z^{2}}}{z {\sqrt {\pi }}} \left( 1 - \frac{1}{2 z^2} + \frac{3}{4 z^4} + \ldots \right)
\end{equation}
where $(2n-1)!!=(2n-1)\cdot (2n-3) \cdots 3\cdot 1$.
Therefore, for $z \rightarrow + \infty$
\begin{equation}\label{eq:AsymErr2}
 \erf(z) = 1 - \erfc(z) \simeq  1 - \frac{e^{-z^{2}}}{z {\sqrt {\pi }}} + \frac{e^{-z^{2}}}{2 z^3 {\sqrt {\pi }}}
\end{equation}
This means that
\begin{equation}\label{eq:AsymErr3}
2 \Phi(z) = 1+ \erf\left(\frac{z}{\sqrt{2}}\right) \xrightarrow[z \to + \infty]{} 2
\end{equation}
Similarly, we obtain
\begin{equation}\label{eq:AsymErr3}
2 \Phi(-z) = 1+ \erf\left(\frac{-z}{\sqrt{2}}\right) = \erfc\left(\frac{-z}{\sqrt{2}}\right) \xrightarrow[z \to + \infty]{} - \frac{\sqrt {2} e^{-\frac{z^{2}}{2}}}{z {\sqrt {\pi }}}
\end{equation}
\begin{equation}\label{eq:AsymErr4}
\ln [ 2 \Phi(-z)] = \ln \left[ 1+ \erf\left(\frac{-z}{\sqrt{2}}\right)\right] = \ln \left[ \erfc\left(\frac{-z}{\sqrt{2}}\right) \right] \simeq \ln \left[ - \frac{\sqrt {2} e^{-\frac{z^{2}}{2}}}{z {\sqrt {\pi }}} \right] \simeq - \frac{z^{2}}{2}
\end{equation}
Then, examining Eq. \eqref{G_s_SN} for $\boldsymbol{\hat{\mu}}^{T}  \boldsymbol{\theta} \geq 0$ and large $r$, we have
\begin{equation}
G_{\overbar{\boldsymbol{X}}}(r, \boldsymbol{\theta}) \simeq \frac{r^2}{2} \boldsymbol{\theta}^{T} \boldsymbol{\Sigma} \boldsymbol{\theta} \, ,
\label{G_s_SN_A}
\end{equation}
while, for $\boldsymbol{\hat{\mu}}^{T}  \boldsymbol{\theta} < 0$ and large $r$,
we can write
\begin{equation}
G_{\overbar{\boldsymbol{X}}}(r, \boldsymbol{\theta}) \simeq - r \boldsymbol{\theta}^{T} \boldsymbol{\hat{\mu}} + \frac{r^2}{2} \boldsymbol{\theta}^{T} \boldsymbol{\Sigma} \boldsymbol{\theta} - \frac{1}{2} \frac{\pi}{2} r^2 \boldsymbol{\theta}^{T} \boldsymbol{\hat{\mu}} \boldsymbol{\hat{\mu}}^{T}  \boldsymbol{\theta}
\simeq \frac{r^{2}}{2}\boldsymbol{\theta}^{T}(\boldsymbol{\Sigma}-\frac{\pi}{2}\boldsymbol{\hat{\mu}}^{T}\boldsymbol{\hat{\mu}}) \boldsymbol{\theta}
\label{G_s_SN_B}
\end{equation}
In Fig. \ref{fig:isopbiskewnormlarger}, we show a graphical representation of the CGF maximization procedure in the case of a bivariate skew-normal random variable.
The red solid lines are the isocurves of $G_{\boldsymbol{X}}(r, \boldsymbol{\theta})$,
and the dashed lines represent the isoprobability contours of $\boldsymbol{X} \sim SN(\boldsymbol{0}, \boldsymbol{\Sigma})$, where the diagonal entries of $\boldsymbol{\Sigma}$ are 1.2 and 0.5143, while the other entries are equal to 0, and $\boldsymbol{\alpha}=(4.365, -1.455)$.

			\begin{figure}[htbp]
\centering
				\includegraphics[width=0.7\textwidth]{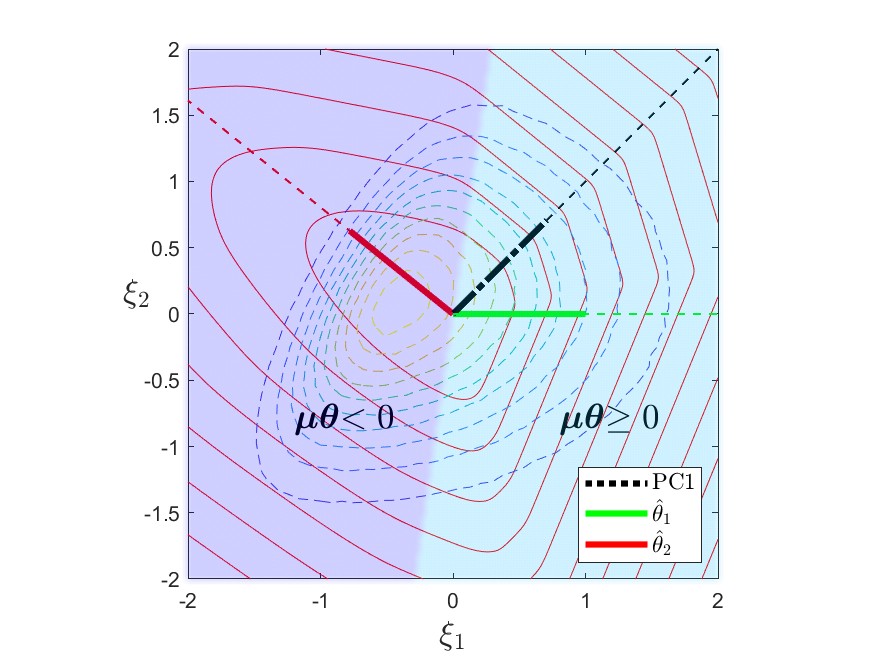}
				\caption{Bivariate skew-normal random variable: the dashed lines represent its isoprobability contours, the red solid lines are the isocurves of
					$G_{\boldsymbol{X}}(r, \boldsymbol{\theta})=-r\boldsymbol{\mu}^{T}\boldsymbol{\theta}+\frac{r}{2}\boldsymbol{\theta}^{T}\boldsymbol{\Sigma}\boldsymbol{\theta}+\ln [2\Phi(\sqrt{\frac{\pi}{2}}r\boldsymbol{\mu}^{T}\boldsymbol{\theta})]$, for large $r$.}
				\label{fig:isopbiskewnormlarger}
			\end{figure}
	
\subsection{Maximizing the non-parametric Cumulant Generating Function}\label{theta_gen}

In this section we address Problem \eqref{MaxCGF} in the case of a generic discrete multivariate random variable $\boldsymbol{X}_{t} = (X_{1,t},\ldots, X_{n,t})$ defined on a discrete state space, where we assume $T$ states of nature, each with probability $\pi_t$ with $t = 1, \ldots , T$.
Therefore, to find the optimal direction $\widehat{\boldsymbol{\theta}}$ that maximizes the non-parametric Cumulant Generating Function, we solve the following optimization problem
\begin{equation}\label{MaxCGFdiscr}
\begin{array}{lll}
\displaystyle\max_{\boldsymbol{\theta}}& \displaystyle G_{\boldsymbol{X}}(r,\boldsymbol{\theta})= \ln\biggl(\sum_{t=1}^{T} \pi_t e^{r \boldsymbol{\theta}^{T} \boldsymbol{X}_{t}}\biggr) \\
\mbox{s.t.}& &\\
& \boldsymbol{\theta}^{T}\boldsymbol{\theta}=1 &
\end{array}
\end{equation}
First, in the following theorem we show that the objective function of this problem is convex.
Let us denote, for the sake of notation,
\begin{equation}\label{eq:disCGF}
  g(\boldsymbol{\xi}) = G_{\boldsymbol{X}}(\boldsymbol{\xi})= \ln\biggl(\sum_{t=1}^{T} \pi_t e^{\boldsymbol{\xi}^{T} \boldsymbol{X}_{t}}\biggr) = G_{\boldsymbol{X}}(r,\boldsymbol{\theta}) = \ln\biggl(\sum_{t=1}^{T} \pi_t e^{r \boldsymbol{\theta}^{T} \boldsymbol{X}_{t}}\biggr) \, ,
\end{equation}
where $\boldsymbol{\xi}=r\boldsymbol{\theta}$.
	\begin{theorem}\label{convfun}
	Let	$g(\boldsymbol{\xi})$ be as in \eqref{eq:disCGF}, where $\boldsymbol{\xi} \in \mathbb{R}^{n}$.
Then, $g(\boldsymbol{\xi})$ is a convex function.
	\end{theorem}
	\begin{proof}
Let $\boldsymbol{\xi}, \boldsymbol{\zeta}\in\mathbb{R}^{n}$, and $\gamma \in [0,1]$.
Thus,
\begin{equation*}
		g(\gamma \boldsymbol{\xi} + (1-\gamma)\boldsymbol{\zeta})= \ln\biggl(\sum_{t=1}^{T} \pi_t e^{(\gamma \boldsymbol{\xi}^{T} +(1-\gamma) \boldsymbol{\zeta}^{T}) \boldsymbol{X}_{t}} \biggr) = \ln \biggl( \sum_{t=1}^{T} \pi_t e^{\gamma \boldsymbol{\xi}^{T} \boldsymbol{X}_{t} } \, \, e^{(1-\gamma) \boldsymbol{\zeta}^{T} \boldsymbol{X}_{t}}\biggr)
\end{equation*}
Now, let $u_{t} = \pi_t e^{\boldsymbol{\xi}^{T} \boldsymbol{X}_{t} }$ and $v_{t} = \pi_t e^{\boldsymbol{\zeta}^{T} \boldsymbol{X}_{t}}$ $\forall t = 1, \ldots , T$, hence
\begin{equation}\label{cgf_uv}
		g(\gamma \boldsymbol{\xi}+(1- \gamma) \boldsymbol{\zeta}) = \ln \biggl(\sum_{t=1}^{T} u_{t}^{\gamma}v_{t}^{1- \gamma}\biggr)
\end{equation}
From H\"{o}lder's inequality, we have
\begin{equation*}
		\sum_{t=1}^{T}\lvert x_{t}y_{t}\rvert\leq \biggl(\sum_{t=1}^{T}\lvert x_{t} \rvert^{p}\biggr)^{\frac{1}{p}}\biggl(\sum_{t=1}^{T}\lvert y_{t} \rvert^{q}\biggr)^{\frac{1}{q}} \, ,
\end{equation*}
where $\boldsymbol{x}, \boldsymbol{y}\in\mathbb{R}^{T}$ and $p,q\in[1, \infty)$, with $\displaystyle \frac{1}{p}+\frac{1}{q}=1$.
		
\noindent
Now, considering $x_{t}=u_{t}^{\gamma}$, $y_{t}=v_{t}^{1-\gamma}$, $\displaystyle \frac{1}{p}=\gamma$ and $\displaystyle \frac{1}{q} = 1-\gamma$, we can apply H\"{o}lder's inequality to Expression \eqref{cgf_uv} as follows
\begin{equation*}
\begin{split}
		g(\gamma \boldsymbol{\xi} + (1-\gamma) \boldsymbol{\zeta}) = \ln \biggl(\sum_{t=1}^{T} u_{t}^{\gamma} v_{t}^{1-\gamma} \biggr) & \leq \ln \biggl( \biggl[ \sum_{t=1}^{T} u_{t}^{\gamma \frac{1}{\gamma}} \biggr]^{\gamma} \biggl[\sum_{t=1}^{T} v_{t}^{1-\gamma \frac{1}{1-\gamma}} \biggr]^{1-\gamma}\biggr)=\\
		& = \ln \biggl( \biggl[ \sum_{t=1}^{T} u_{t} \biggr]^{\gamma} \biggl[ \sum_{t=1}^{T} v_{t} \biggr]^{1-\gamma}\biggr)=\\
		&=\gamma \ln \biggl(\sum_{t=1}^{T} u_{t} \biggr)+(1-\gamma) \ln \biggl( \sum_{t=1}^{T} v_{t} \biggr)=\\
		&=\gamma g(\boldsymbol{\xi}) + (1-\gamma) \,  g(\boldsymbol{\zeta})
\end{split}
\end{equation*}
which completes the proof.
\end{proof}
Therefore, Problem \eqref{MaxCGFdiscr} consists in maximizing a convex function on an $n$-sphere, that is a nonconvex set.
However, thanks to the following proposition we can relax
such an $n$-sphere into an $n$-ball.
\begin{proposition}
  Problem \eqref{MaxCGFdiscr} is equivalent to
\begin{equation}\label{MaxCGFdiscrBall}
\begin{array}{lll}
\displaystyle\max_{\boldsymbol{\theta}}& \displaystyle G_{\boldsymbol{X}}(r,\boldsymbol{\theta})= \ln\biggl(\sum_{t=1}^{T} \pi_t e^{r \boldsymbol{\theta}^{T} \boldsymbol{X}_{t}}\biggr) \\
\mbox{s.t.}& &\\
& \boldsymbol{\theta}^{T}\boldsymbol{\theta} \leq 1 &
\end{array}
\end{equation}
\end{proposition}
\begin{proof}
  Since the $n$-ball $\boldsymbol{\theta}^{T}\boldsymbol{\theta} \leq 1$ can be seen as the convex hull of the $n$-sphere $\boldsymbol{\theta}^{T}\boldsymbol{\theta} = 1$, we can apply Theorem 32.2 of \cite{rockafellar1970convex},
  which concludes the proof.
\end{proof}
Thus, 
Problem \eqref{MaxCGFdiscrBall} is a concave programming problem, which is NP hard \citep[][]{benson1995concave}.
Indeed, the problem of globally maximizing a convex function on a convex set may have many local minima, hence finding the global maximum is a computationally difficult problem, and several approaches have been proposed in the literature to address this problem    \citep[see, e.g.,][]{pardalos1986methods}.

\subsection{A Heuristic for maximizing CGF }\label{sec:CGFMaximum}

Similar to \cite{bernacchia2008detecting2},
to solve Problem \eqref{MaxCGFdiscrBall} we use a multistart method, which consists in generating several random initial points belonging to the $n$-sphere to determine a point in a neighborhood of a global maximum, and in using a local maximizer to efficiently determine a global maximum.
We point out that we randomly generate initial points on the $n$-sphere, because, if there exists at least a global maximum of CGF on an $n$-ball, this belongs to its frontier \citep[see Corollary 32.3.1 of][]{rockafellar1970convex}.
Furthermore, since the constraint of Problem \eqref{MaxCGFdiscrBall} defines a convex set, we use the projected gradient algorithm as the local maximizer \citep[see][]{trendafilov2006projected}.
In Table \ref{PseudoCode}
we summarize the CGF maximization procedure.
\begin{table}[htbp]
  \centering
\scalebox{0.90}{
\begin{tcolorbox}
	\begin{description}[noitemsep]
\item[\Cag{1.}] Generate $N$ starting points (directions) $\boldsymbol{\theta}_{0}^{(j)}$ with $j=1, \ldots, N$  belonging to the unit $n$-sphere \citep[see][]{knuth2014art}. Fix the step size $\delta_{i}=\delta=\frac{1}{r}$ and the tolerance $\epsilon$ sufficiently small \citep[see][]{nocedal1999numerical}.
\hspace{-10cm}
\item[\Cag{2.}] \textbf{for} $j=1, \ldots, N$
\item[\Cag{3.}] \quad Set $i=0$
\item[\Cag{4.}] \quad \textbf{while} $\norm{\boldsymbol{\theta}_{i+1}^{(j)} -\boldsymbol{\theta}_{i}^{(j)}}> \epsilon$
\item[\Cag{5.}] \qquad Compute the ascent direction as $\nabla_{\boldsymbol{\theta}_{i}} G_{\boldsymbol{X}}(r, \boldsymbol{\theta}_{i}^{(j)})$
\item[\Cag{6.}] \qquad Update $\boldsymbol{\theta}_{i+1}^{(j)} = \boldsymbol{\theta}_{i}^{(j)} + \delta_{i} \nabla_{\boldsymbol{\theta}_{i}} G_{\boldsymbol{X}}(r, \boldsymbol{\theta}_{i}^{(j)})$
    %
\item[\Cag{7.}] \qquad Project by rescaling $\boldsymbol{\theta}_{i+1}^{(j)}$ to unit norm, i.e., $\boldsymbol{\theta}_{i+1}^{(j)} = \frac{ \boldsymbol{\theta}_{i+1}^{(j)}}{ \norm{\boldsymbol{\theta}_{i+1}^{(j)}}}$
\item[\Cag{8.}] \qquad Update $i=i+1$
\item[\Cag{9.}] \quad \textbf{end while}
\item[\Cag{10.}] \textbf{end for}
\end{description}
\end{tcolorbox}
}
  \caption{Pseudocode of the CGF maximization procedure}\label{PseudoCode}
\end{table}
%

\noindent
More precisely, in Expression \eqref{eq:disCGF}, we assume $\pi_{t}=\frac{1}{T}$ for $t=1, \cdots, T$, and, therefore, the CGF of a discrete multivariate variable $\boldsymbol{X}_{t}=(X_{1,t},\ldots, X_{n,t})$ becomes
	\begin{equation}\label{eq:CGFSampleEstimate}
	G_{\boldsymbol{X}}(r, \boldsymbol{\theta})=\ln \displaystyle\frac{1}{T} \sum_{t=1}^{T} e^{r \boldsymbol{\theta}^{T} \boldsymbol{X}_{t}} \, .
	\end{equation}
For a fixed $r$, we set $N=1000$ and $\delta_{i}=\frac{1}{r}$, where $N$ is the number of starting points of the multistart heuristic and $\delta_{i}$ is the step size of the projected gradient algorthm (see Table \ref{PseudoCode}).
Furthermore, from Expression \eqref{eq:CGFSampleEstimate}, for a given $r$, the gradient of $G_{\boldsymbol{X}}(r, \boldsymbol{\theta})$ is
\begin{equation}\label{eq:gradient}
  \nabla_{\boldsymbol{\theta}} G_{\boldsymbol{X}}(r, \boldsymbol{\theta})= r \displaystyle \frac{\frac{1}{T} \sum_{t=1}^{T} \boldsymbol{X}_{t} e^{r \boldsymbol{\theta}^{T} \boldsymbol{X}_{t} }} {\frac{1}{T} \sum_{t=1}^{T} e^{r \boldsymbol{\theta}^{T} \boldsymbol{X}_{t}}}
\end{equation}
Therefore, in Step 6 of the pseudocode in Table \ref{PseudoCode}, the iterative scheme to find a local optimum is defined by
	\begin{equation}
	\label{Alg1}
	\boldsymbol{\theta}_{i+1} - \boldsymbol{\theta}_{i} = \delta_{i} \nabla_{\boldsymbol{\theta}_{i}} G_{\boldsymbol{X}}(r, \boldsymbol{\theta}_{i}) = \displaystyle \frac{ \sum_{t=1}^{T} \boldsymbol{X}_{t} e^{r \boldsymbol{\theta}_{i}^{T} \boldsymbol{X}_{t}}} {\sum_{t=1}^{T} e^{r \boldsymbol{\theta}_{i}^{T} \boldsymbol{X}_{t}}} \, ,
	\end{equation}
while, in Step 7, $\boldsymbol{\theta}_{i+1}$ is normalized to one.

\noindent
Thus applying the CGF maximization procedure to Problem  \eqref{MaxCGFdiscrBall}, we can obtain the local maxima $\boldsymbol{\hat{\theta}}=\boldsymbol{\hat{\theta}}(r)$, which, except for symmetric distributions (see Section \ref{sec:Normal}),
depend on $r$.
As discussed in Section \ref{sec:cumulantsoptdir},
if $r$ (the distance between data points and the center) is small, then the CGF maximization procedure essentially picks the first principal component of the classical PCA.
Whereas, if $r$ is large, then the optimal directions $\boldsymbol{\hat{\theta}}$ maximizing $G_{\boldsymbol{X}}(r, \boldsymbol{\theta})$ strongly depend on
the higher-order cumulants.
On the other hand, higher values of $r$ produce a less accurate estimate of the function $G(r, \boldsymbol{\theta})$, since the data sample is finite and a few outliers could heavily influence the higher-order cumulants.
Therefore, when setting the value of $r$, one must consider the trade-off between finding $\boldsymbol{\hat{\theta}}(r)$ with large $r$ (thus involving higher-order cumulants in the procedure) and obtaining an accurate estimate of \eqref{eq:CGFSampleEstimate}.
As suggested by \cite{bernacchia2008detecting2},
we set a value $\bar{r}$ such that the CGF estimation error is limited.
We identify such an error by means of the relative variance of the CGF estimate, $\varepsilon_{G}^{2}$, that, as shown in Appendix \ref{RelErr}, assuming i.i.d. normally distributed random vectors $\boldsymbol{X}_{t}\sim N_{n}(\boldsymbol{0}, \boldsymbol{\Sigma})$ $\forall t =1, \ldots, T$, is as follows
\begin{equation}\label{eq:RelVarCGF_Taylor_First3a}
  \varepsilon_{G}^{2} = \frac{\V[G]}{(\E[G])^2} \simeq \frac{4}{T} \frac{e^{ \bar{r}^2 \lambda_1 } -1}{\bar{r}^4 ( \lambda_1 )^2} \, ,
\end{equation}
where $T$ is the number of data points and $\lambda_1$ is the largest eigenvalue computed by standard PCA.
In the experimental analysis, $\bar{r}$ is found through \eqref{eq:RelVarCGF_Taylor_First3a} by setting $\varepsilon_{G}=10\%$.

\section{Outlier detection methodologies}\label{sec:outdet}

In this section, we introduce the proposed outlier detection algorithm for multivariate data, named the MaxCGF algorithm, consisting in finding outliers in univariate projections of such data. Our approach is compared with two other projection-based methods, i.e., one developed by \cite{pena2001multivariate, pena2007combining} and the other developed by \cite{domino2020multivariate}. Their main difference relies on the selected directions onto which the data are projected.
More precisely, \cite{pena2001multivariate, pena2007combining} consider the directions for which the kurtosis of the projected data shows the highest and lowest values. \cite{domino2020multivariate} takes into account the projection directions for which a multivariate series exhibits the highest absolute value of the 4$^{th}$ order cumulant. We propose a more general approach that is based on the projection directions that maximize the Cumulant Generating Function of the multivariate data.
In the following we report the main steps of our procedure.
\begin{enumerate}
	\item Preprocess the original data $\boldsymbol{X}$, yielding centered data $\boldsymbol{Y}$.
Note that while both \cite{pena2001multivariate, pena2007combining} and \cite{domino2020multivariate} standardize $\boldsymbol{X}$, we only center them.
\item Find the directions that maximize CGF of $\boldsymbol{Y}$.
\item Compute the univariate projection $\boldsymbol{Z}$ of $\boldsymbol{Y}$ in the directions identified in Step 2.
\item Remove outliers. Note that to identify whether a generic element $z_t$ of $\boldsymbol{Z}=\{z_t \}_{t=1, \ldots, T}$ is an outlier, similar to \cite{pena2001multivariate,pena2007combining} and \cite{domino2020multivariate}, we compute the following quantity
\begin{equation}
		\label{unimeas}
q_{t} = \frac{|z_{t} - \mbox{median} (\boldsymbol{Z})|}{MAD( \boldsymbol{Z})} , \quad t = 1,\ldots,T \, ,
\end{equation}
where
$T$ is the length of the time series, and $MAD(\boldsymbol{Z})$ represents the median absolute deviation of $\boldsymbol{Z}$.
Then, a generic outcome $z_t$ is classified as an outlier if its corresponding $q_{t}$ exceeds a fixed threshold $\beta$.
Such a threshold is selected to cover the whole detection range both in terms of True Positive Rate and in terms of False Positive Rate (for more details, see Section \ref{sec:expana}).
In order to determine the optimal $\beta$, we consider equally spaced values of $\beta$ ranging from 0.5 to 10 with a step of 0.25.
\item Repeat Steps 2, 3 and 4 until the kurtosis of the univariate projection $\boldsymbol{Z}$ increases \citep[see][]{domino2020multivariate}.
\end{enumerate}
In Table \ref{PseudoCodeAlgOut} we summarize the MaxCGF algorithm (pseudocode) for outlier detection.
\begin{table}[htbp]
\centering
\scalebox{0.95}{
\begin{tcolorbox}
\begin{description}[noitemsep]
\item[\Cag{1.}] Fix the threshold $\beta$;
\hspace{-10cm}
\item[\Cag{2.}] Center data, i.e., $\boldsymbol{Y}^{(0)} = \boldsymbol{X} - \boldsymbol{\mu}$;
\item[\Cag{3.}] Find the $N$ directions $\hat{\boldsymbol{\theta}}_{1}^{(0)}, \hat{\boldsymbol{\theta}}_{2}^{(0)}, \ldots, \hat{\boldsymbol{\theta}}_{N}^{(0)}$, maximizing CGF of $\boldsymbol{Y}^{(0)}$ (see Table \ref{PseudoCode});
\item[\Cag{4.}] \textbf{for} $j=1:N$    (i.e., for each direction  $\hat{\boldsymbol{\theta}}_{j}^{(0)}$)
\item[\Cag{5.}] \quad Set $i=0$;
\item[\Cag{6.}] \quad Project $\boldsymbol{Y}^{(0)}$ on $\hat{\boldsymbol{\theta}}_{j}^{(0)}$, thus obtaining the vector $\boldsymbol{Z}_{j}^{(0)}=\boldsymbol{Y}^{(0)} \hat{\boldsymbol{\theta}}_{j}^{(0)}$;
\item[\Cag{7.}] \quad Compute $\mathrm{Kur}_{0}=\mathbb{K}\mathrm{ur}(\boldsymbol{Z}_{j}^{(0)})$, i.e., the kurtosis of $\boldsymbol{Z}_{j}^{(0)}$;
\item[\Cag{8.}] \qquad \textbf{while} $\mathrm{Kur}_{i} < \mathrm{Kur}_{i-1}$ or $i=0$
\item[\Cag{9.}] \qquad \, Compute the vector $\boldsymbol{q}=\{q_t \}_{t=1, \ldots, T}$ as in \eqref{unimeas} for $\boldsymbol{Z}_{j}^{(i)}$;
\item[\Cag{10.}] \qquad Remove the outliers when $q_{t}>\beta$, thus
				 obtaining $\boldsymbol{Y}^{(i+1)}$;
\item[\Cag{11.}] \qquad From $\boldsymbol{Y}^{(i+1)}$, compute $\hat{\boldsymbol{\theta}}_{j}^{(i+1)}$ by using the algorithm in Table \ref{PseudoCode};
\item[\Cag{12.}] \qquad Project $\boldsymbol{Y}^{(i+1)}$ on $\hat{\boldsymbol{\theta}}_{j}^{(i+1)}$, thus obtaining the vector $\boldsymbol{Z}_{j}^{(i+1)}= \boldsymbol{Y}^{(i+1)} \hat{\boldsymbol{\theta}}_{j}^{(i+1)}$;
\item[\Cag{13.}] \qquad Compute $\mathrm{Kur}_{i+1}=\mathbb{K}\mathrm{ur}(\boldsymbol{Z}_{j}^{(i+1)})$, i.e., the kurtosis of $\boldsymbol{Z}_{j}^{(i+1)}$;
\item[\Cag{14.}] \qquad Update $i=i+1$
\item[\Cag{15.}] \quad \, \textbf{end while}
\item[\Cag{16.}] \quad $\boldsymbol{Y}^{(0)} = \boldsymbol{Y}^{(i+1)}$, $\hat{\boldsymbol{\theta}}_{j}^{(0)} = \hat{\boldsymbol{\theta}}_{j}^{(i+1)}$
\item[\Cag{17.}] \textbf{end for} 
			\end{description}
		\end{tcolorbox}
	}
	\caption{Pseudocode of the MaxCGF algorithm}\label{PseudoCodeAlgOut}
\end{table}

\noindent
In the next section, we compare the outlier detection ability of our method with that of two alternative methods proposed by \cite{pena2001multivariate,pena2007combining} and \cite{domino2020multivariate}.

\section{Experimental analysis}\label{sec:expana}

In this section, we provide a thorough empirical analysis, where the three outlier detection procedures discussed in Section \ref{sec:outdet} are tested and compared both using simulated data (see Section \ref{subsec:simdata}) and financial real-world data (see Section \ref{subsec:realdata}).
Performances are evaluated using the Receiver Operating Characteristic (ROC) curves, which are built by plotting the True Positive Rate (TPR) vs. the False Positive Rate (FPR) of each methodology for different levels of the threshold $\beta$.
TPR is the rate of truly detected outliers w.r.t. all possible outliers, and FPR (i.e., the type I error) is the rate of falsely detected outliers w.r.t. all the non-outlier data.
Furthermore, we also provide two other performance measures related to the ROC curve, namely the Area Under the Curve (AUC) and Youden's J Statistic (YJS).
AUC is the area underneath the ROC curve: if a method perfectly distinguishes between outlier and non-outlier data, then its AUC is equal to 1; if, on the other hand, AUC is less than or equal to 0.5, then this measure is uninformative, i.e., there is no difference in performance between the analyzed method and one that relies on random choices.
YJS represents the difference between TPR and FPR, and again, the higher, the better.
The Best Cutoff Value (BCV) is the largest YJS, and $\beta^{*}$ is the threshold value corresponding to BCV.
Although $\beta^{*}$ is not a performance measure in the strict sense, it provides information about the practical use of algorithms and allows one to identify whether there is an optimal range of $\beta$ for which algorithms work best.

\noindent
Note that in our experiments we report the computational results for all the three methods only in three instances (standard normal, normal, and real-world data from the Dow Jones market).
In the remaining cases, we report only the computational results of two methods since the Pe\~na-Prieto algorithm does not provide reasonable results, and the generated ROC curves do not show a nondecreasing trend at some points, or yield AUC values less than 0.5.

\noindent
We implemented all the experiments on a workstation with Intel(R) Xeon(R) CPU (E5-2623 v4, 2.6 GHz, 64 Gb RAM) under MS Windows 10, using MATLAB 9.12.0.

\subsection{Computational results for simulated data}\label{subsec:simdata}

Artificial simulated returns are drawn from four different distributions: standard normal (Section \ref{subsubsec:StdNormExpAna}), normal (Section \ref{subsubsec:NormExpAna}), skew-normal (Section \ref{subsubsec:SkewNormExpAna}) and Student's t (Section \ref{subsubsec:StudtExpAna}), with $n=30$ (the number of marginals, i.e., assets) and $T=500$ (the number of scenarios), which corresponds to 2 financial years.
Following \cite{domino2020multivariate}, the outliers matrix $\boldsymbol{O}_{(0.1 T  \times 0.5 n)}$ has dimensions $0.1 \ T \times 0.5 \ n$.
Below, we report the main steps for generating outliers and substituting them into the ordinary data.
\begin{enumerate}
	\item Compute the sample covariance matrix $\boldsymbol{\Sigma}$ from a real-world dataset (here, the Dow Jones dataset from 07/2004 to 07/2006, see Section \ref{subsec:realdata}).
	\item Generate the ordinary data $\boldsymbol{X}_{(T \times n)}$ with $T$ scenarios from an $n$-dimensional random vector using $\boldsymbol{\Sigma}$
 (in the standard normal case, $\boldsymbol{\Sigma}=\boldsymbol{I}$);
generate the outlier data $\boldsymbol{O}_{(0.1 T  \times 0.5 n)}$ using $\boldsymbol{\Sigma}$ multiplied by 15.
\item Substitute the elements of $\boldsymbol{O}_{(0.1 T  \times 0.5 n)}$ in $\boldsymbol{X}_{(T \times n)}$.
More precisely, randomly select $n/2$ columns (with $0.1 T$ elements) of $\boldsymbol{X}_{(T \times n)}$ and substitute them with the columns of $\boldsymbol{O}_{(0.1 T  \times 0.5 n)}$.
\end{enumerate}

\subsubsection{The standard normal random vector case}\label{subsubsec:StdNormExpAna}

For this experimental case, the ordinary and outlier data are $\boldsymbol{X}_{(T \times n)} \sim N(\boldsymbol{0}, \boldsymbol{I})$ and $\boldsymbol{O}_{(0.1 T  \times 0.5 n)} \sim N(\boldsymbol{0}, 15 \boldsymbol{I})$.

\noindent
In Table \ref{tab:StdNormperf}, we report the computational results obtained by using the three methods described in Section \ref{sec:outdet}.
More precisely, the Pe\~na-Prieto method obtains an almost perfect score of 0.9867 both for the AUC and the BCV performance measures.
The Domino algorithm shows the lowest performance w.r.t. the other two methods, both in terms of AUC ($=0.9316$) and BCV ($=0.7489$).
The MaxCGF method almost achieves the highest values in terms of accuracy and clearly shows a significant advantage in terms of computational burden. Its running time is less than 1 minute, compared with 15 minutes and more than 1 hour spent by the Domino and Pe\~na-Prieto methods, respectively.
\begin{table}[htbp]
	\centering
	\vspace{5pt}
	{\renewcommand{\arraystretch}{1.1}
		\hspace*{-10pt}{\small{
				\begin{tabular}{*{5}{!{\vrule width 0.9pt}p{3cm}}!{\vrule width 1.5pt}}
					\hline
					\textbf{Method} & \textbf{AUC} & \textbf{BCV} & \textbf{Time (min.)} & $\boldsymbol{\beta}^{*}$ \\ \hhline{|-|-|-|-|-|}
					\textbf{MaxCGF} & \cellcolor[rgb]{ 1,  1,  0}$0.9843$ & \cellcolor[rgb]{ 1,  1,  0}$0.9533$ & \cellcolor[rgb]{ 0,  .69,  .314}$0.6032$ & 3.25 \\
					\hhline{|-|-|-|-|-|}
					\textbf{Domino} & \cellcolor[rgb]{ 1,  0.44,  0.37} $0.9316$ & \cellcolor[rgb]{ 1,  0.44,  0.37}$0.7489$ & \cellcolor[rgb]{ 1,  1,  0}$21.6774$ & 4.25 \\
					\hhline{|-|-|-|-|-|}
					\textbf{Pe\~{n}a-Prieto} & \cellcolor[rgb]{ 0,  .69,  .314} $0.9867$ & \cellcolor[rgb]{ 0,  .69,  .314}$0.9867$ & \cellcolor[rgb]{ 1,  0.44,  0.37}$66.6850$ & 2.25 \\
					
					\hline
				\end{tabular}}}
			}
			\caption{Performances for the standard normal dataset. We mark in green the best, in yellow the intermediate, and in red the worst results.}
			\label{tab:StdNormperf}%
		\end{table}%

\noindent
Figure \ref{fig:ROCStdNorm} reports the ROC curves of the three methods, which, as already noted, show high outlier detection abilities in this experiment.
\begin{figure}[htbp]
	\centering
	\includegraphics[width=0.6\textwidth]{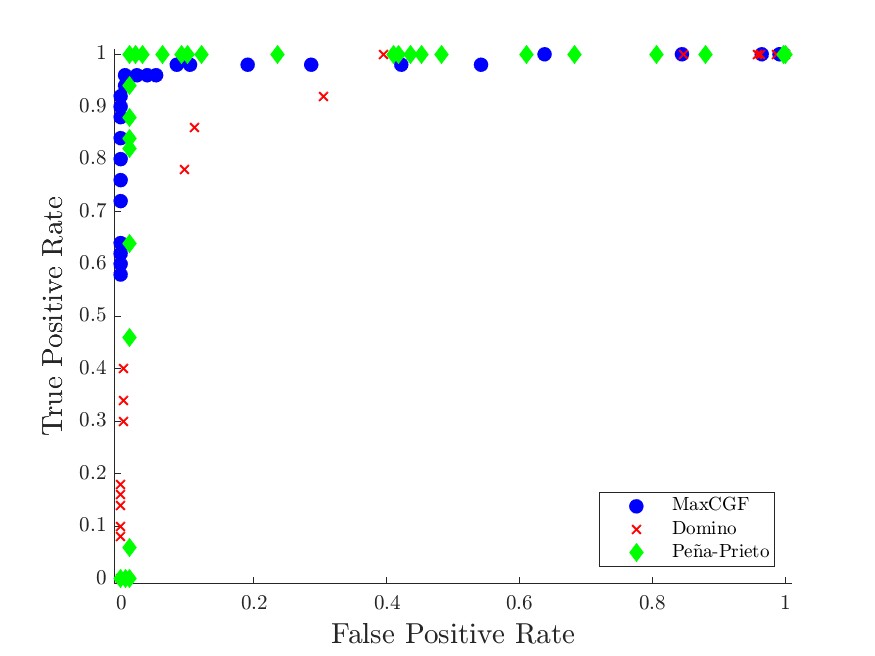}
	\caption{ROC curves of the three algorithms for the standard normal market.}
	\label{fig:ROCStdNorm}
\end{figure}

\subsubsection{The normal random vector case}\label{subsubsec:NormExpAna}

In this case, the ordinary and outlier data are $\boldsymbol{X}_{(T \times n)} \sim N(\boldsymbol{0}, \boldsymbol{\Sigma})$ and $\boldsymbol{O}_{(0.1 T  \times 0.5 n)} \sim N(\boldsymbol{0}, 15 \boldsymbol{\Sigma})$.

\noindent The empirical results of the ROC analysis are shown in Figure \ref{fig:ROCNorm} and in Table \ref{tab:Normperf}.
\begin{figure}[htbp]
	\centering
	\includegraphics[width=0.6\textwidth]{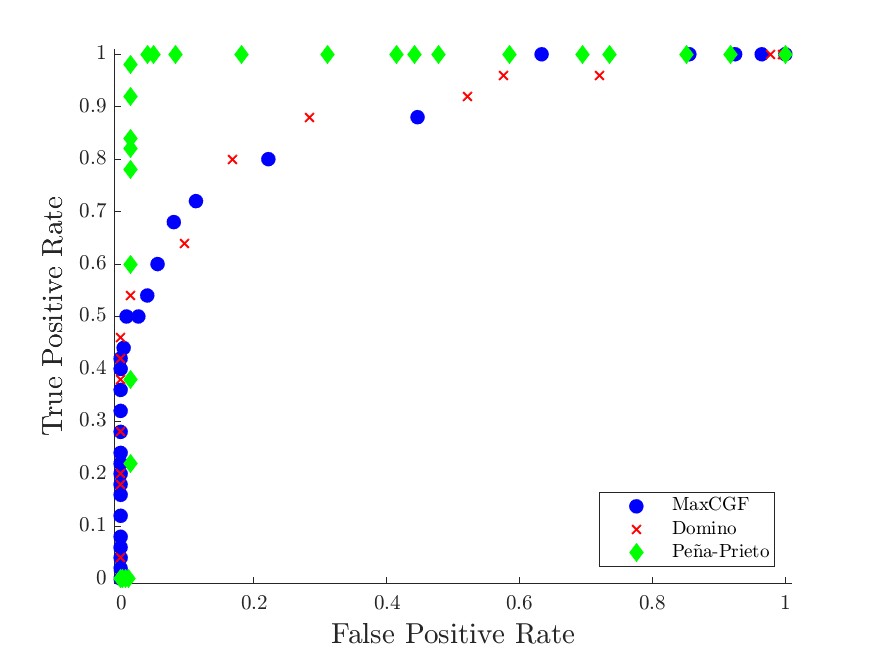}
	\caption{ROC curves of the three algorithms for the normal market.}
	\label{fig:ROCNorm}
\end{figure}
\begin{table}[htbp]
	\centering
	\vspace{5pt}
	{\renewcommand{\arraystretch}{1.1}
		\hspace*{-10pt}{\small{
				\begin{tabular}{*{5}{!{\vrule width 0.9pt}p{3cm}}!{\vrule width 1.5pt}}
					\hline
					\textbf{Method} & \textbf{AUC} & \textbf{BCV} & \textbf{Time (min.)} & $\boldsymbol{\beta}^{*}$ \\ \hhline{|-|-|-|-|-|}
					\textbf{MaxCGF} & \cellcolor[rgb]{ 1, 1, 0} $0.8811$ & \cellcolor[rgb]{ 1,  0.44,  0.37}$0.6067$ & \cellcolor[rgb]{ 0,  .69,  .314}$0.9592$ & 8.00 \\
					 \hhline{|-|-|-|-|-|}
					\textbf{Domino} & \cellcolor[rgb]{ 1,  0.44,  0.37} $0.8809$ & \cellcolor[rgb]{ 1,  1,  0}$0.6311$ & \cellcolor[rgb]{ 1,  1,  0}$14.3293$ & 7.50 \\
					\hhline{|-|-|-|-|-|}
					\textbf{Pe\~{n}a-Prieto} & \cellcolor[rgb]{ 0,  .69,  .314} $0.9847$ & \cellcolor[rgb]{ 0,  .69,  .314}$0.9644$ & \cellcolor[rgb]{ 1,  0.44,  0.37}$67.2382$ & 1.75 \\
					
					\hline
				\end{tabular}}}
			}
			\caption{Performances for the normal dataset.}
			\label{tab:Normperf}%
		\end{table}%
Here, the ranking of the three algorithms is similar to that of the standard normal dataset.
Again, the Pe\~na-Prieto algorithm outperforms the other two both in terms of AUC and BCV, while the MaxCGF algorithm is the most efficient; however, the Domino algorithm yields the second best BCV.
It is noteworthy that, compared to the standard normal case, the introduction of a correlation structure seems to worsen the results for the MaxCGF and Domino algorithms, while that of Pe\~na-Prieto does not experience noticeable changes.

\subsubsection{The skew-normal random vector case}\label{subsubsec:SkewNormExpAna}

For this experiment, the ordinary and outlier data are $\boldsymbol{X}_{(T \times n)} \sim SN(\boldsymbol{0}, \boldsymbol{\Sigma}, \boldsymbol{\alpha})$ and $\boldsymbol{O}_{(0.1 T  \times 0.5 n)} \sim SN(\boldsymbol{0}, 15 \boldsymbol{\Sigma}, \boldsymbol{\alpha})$, where $\boldsymbol{\alpha}$ is a vector whose elements are uniformly distributed in the interval $[-1,4]$.
		\begin{table}[htbp]
			\centering
			\vspace{5pt}
			{\renewcommand{\arraystretch}{1.1}
				\hspace*{-10pt}{
\scalebox{0.9}{
						\begin{tabular}{*{5}{!{\vrule width 0.9pt}p{3cm}}!{\vrule width 1.5pt}}
							\hline
							\textbf{Method} & \textbf{AUC} & \textbf{BCV} & \textbf{Time (min.)} & $\boldsymbol{\beta}^{*}$ \\ \hhline{|-|-|-|-|-|}
							\textbf{MaxCGF} & \cellcolor[rgb]{ 0,  .69,  .314}$0.9140$ & \cellcolor[rgb]{ 0,  .69,  .314}$0.6911$ & \cellcolor[rgb]{ 0,  .69,  .314}$1.0387$ & 7.25 \\ \hhline{|-|-|-|-|-|}
							\textbf{Domino} & \cellcolor[rgb]{ 1,  0.44,  0.37}$0.8896$ & \cellcolor[rgb]{ 1,  0.44,  0.37}$0.6467$ & \cellcolor[rgb]{ 1,  0.44,  0.37}$15.8041$ & 6.50 \\
							\hhline{|-|-|-|-|-|}
							\textbf{Pe\~{n}a-Prieto} & - & - & - & - \\
							\hline
						\end{tabular}}}
					}
					\caption{Performances for the skew-normal dataset.}
					\label{tab:SkewNormperf}%
				\end{table}%
\begin{figure}[htbp]
	\centering
	\includegraphics[width=0.6\textwidth]{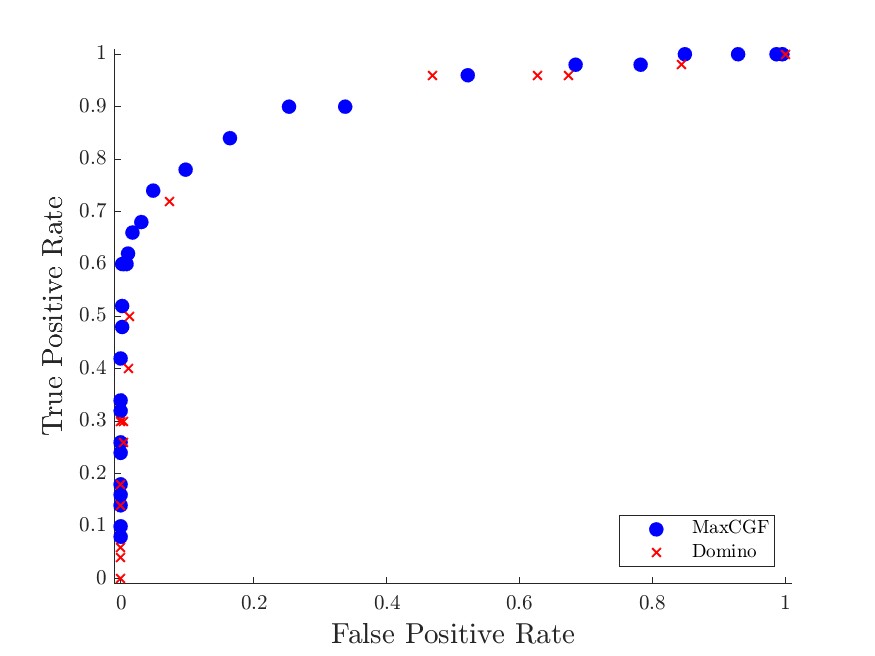}
\caption{ROC curves of the MaxCGF and Domino algorithms for the skew-normal market.}
	\label{fig:ROCSkewNorm}
\end{figure}

\noindent
As shown in Table \ref{tab:SkewNormperf} and in Figure \ref{fig:ROCSkewNorm}, the MaxCGF algorithm provides, in the case of asymmetric distribution, the highest AUC (=0.9140) and BCV (=0.6467), and it is the most efficient approach. The Domino method also performs well, although lower than the MaxCGF method. Its AUC is 0.8896, and its BCV is 0.6467, while its running time is 15 minutes compared with 1 minute spent by the MaxCGF method.

\subsubsection{The Student's t random vector case}\label{subsubsec:StudtExpAna}

Here, the ordinary and outlier data follow an $n$-variate Student's t distribution, $\boldsymbol{X}_{(T \times n)} \sim St(\boldsymbol{0}, \boldsymbol{\Sigma}, \nu)$ and $\boldsymbol{O}_{(0.1 T  \times 0.5 n)} \sim St(\boldsymbol{0}, 15 \boldsymbol{\Sigma}, \nu)$, where $\nu$ is set to 5, 10, 100, 1000.
Clearly, when $\nu$ is sufficiently high, the Student's t random vector approaches the normal one.

\noindent
Before testing the three methods described in Section \ref{sec:outdet}, we examine the outlier detection procedure by using the projection directions provided by the CGF maximization and by the classical PCA.
\begin{table}[htbp]
	\centering
	{\renewcommand{\arraystretch}{1.1}
		\hspace*{-10pt}{
\scalebox{0.9}{
				\begin{tabular}{*{5}{!{\vrule width 0.9pt}p{3cm}}!{\vrule width 1.5pt}}
					\hline
					$\nu$ & \textbf{Method} & \textbf{AUC} & \textbf{BCV} & \textbf{Time (min.)} \\ \hhline{|-|-|-|-|}
					\multirow{ 2}{*}{5}& \textbf{Classical PCA} & \cellcolor[rgb]{ 1,  0.44,  0.37}$0.7757$ & \cellcolor[rgb]{ 1,  0.44,  0.37}$0.4467$ & \cellcolor[rgb]{ 0,  .69,  .314}$0.8389$ \\
					& \textbf{MaxCGF} & \cellcolor[rgb]{ 0,  .69,  .314}$0.8544$ & \cellcolor[rgb]{ 0,  .69,  .314}$0.5933$ & \cellcolor[rgb]{ 1,  0.44,  0.37}$22.5747$  \\
					\hhline{|=|=|=|=|=|}
					\multirow{ 2}{*}{10}& \textbf{Classical PCA} & \cellcolor[rgb]{ 1,  0.44,  0.37}$0.7988$ & \cellcolor[rgb]{ 1,  0.44,  0.37}$0.4089$ & \cellcolor[rgb]{ 0,  .69,  .314}$0.7955$ \\
					& \textbf{MaxCGF} &    \cellcolor[rgb]{ 0,  .69,  .314}$0.8333$ & \cellcolor[rgb]{ 0,  .69,  .314}$0.5356$ & \cellcolor[rgb]{ 1,  0.44,  0.37}$10.0730$  \\
					\hhline{|=|=|=|=|=|}
					\multirow{ 2}{*}{100}& \textbf{Classical PCA} & \cellcolor[rgb]{ 1,  0.44,  0.37}$0.8434$ & \cellcolor[rgb]{ 1,  0.44,  0.37}$0.5600$ & \cellcolor[rgb]{ 0,  .69,  .314}$0.7952$ \\
					& \textbf{MaxCGF} &     \cellcolor[rgb]{ 0,  .69,  .314}$0.8525$ & \cellcolor[rgb]{ 0,  .69,  .314}$0.6311$ & \cellcolor[rgb]{ 1,  0.44,  0.37}$21.8163$  \\
					\hhline{|=|=|=|=|=|}
					\multirow{ 2}{*}{1000}& \textbf{Classical PCA} & \cellcolor[rgb]{ 0,  .69,  .314}$0.8504$ & \cellcolor[rgb]{ 0,  .69,  .314}$0.5756$ & \cellcolor[rgb]{ 0,  .69,  .314}$0.7755$ \\
					& \textbf{MaxCGF} & \cellcolor[rgb]{ 1,  0.44,  0.37}$0.8409$ & \cellcolor[rgb]{ 1,  0.44,  0.37}$0.5489$ & \cellcolor[rgb]{ 1,  0.44,  0.37}$17.5729$ \\
					\hhline{|-|-|-|-|-|}
					
					\hline
				\end{tabular}}}
			}
			\caption{Outlier detection performances using the CGF maximization and the classical PCA for the Student's t random vector.
}
			\label{tab:Studtthetaperf}%
		\end{table}%
As shown in Table \ref{tab:Studtthetaperf}, the latter is much more efficient in terms of running time, but, except for $\nu=1000$, the MaxCGF method is more precise, as shown by the values of AUC and BCV. Furthermore, the accuracy of the classical PCA tend to increase with $\nu$, namely when the simulated data tend to have a normal distribution. Figure \ref{fig:ROCStudttheta} exhibits the ROC curves of such experiments.
\begin{figure}[htbp]
	\centering
	\subfigure[$\nu=5$]
	{
		\includegraphics[width=7cm]{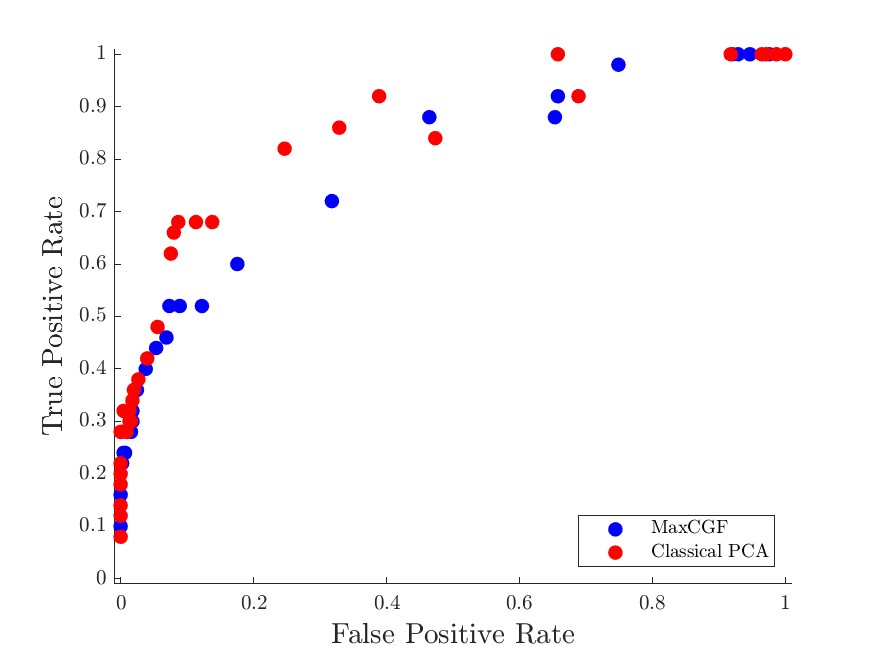}
		\label{fig:ROCStudtthetanu5}
	}
	\hfill
	\subfigure[$\nu=10$]
	{
		\includegraphics[width=7cm]{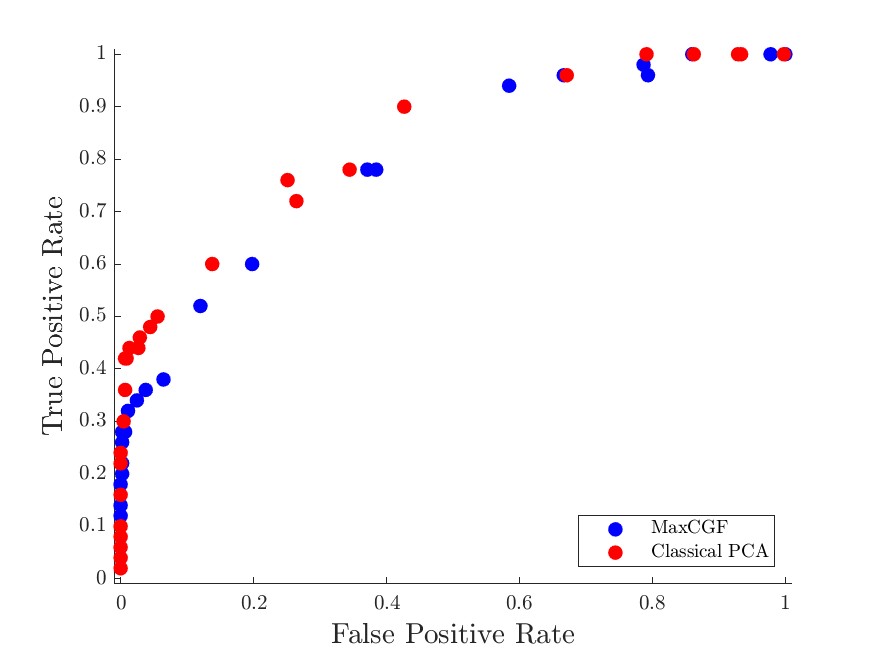}
		\label{fig:ROCStudtthetanu10}
	}
	\\
	\subfigure[$\nu=100$]
	{
		\includegraphics[width=7cm]{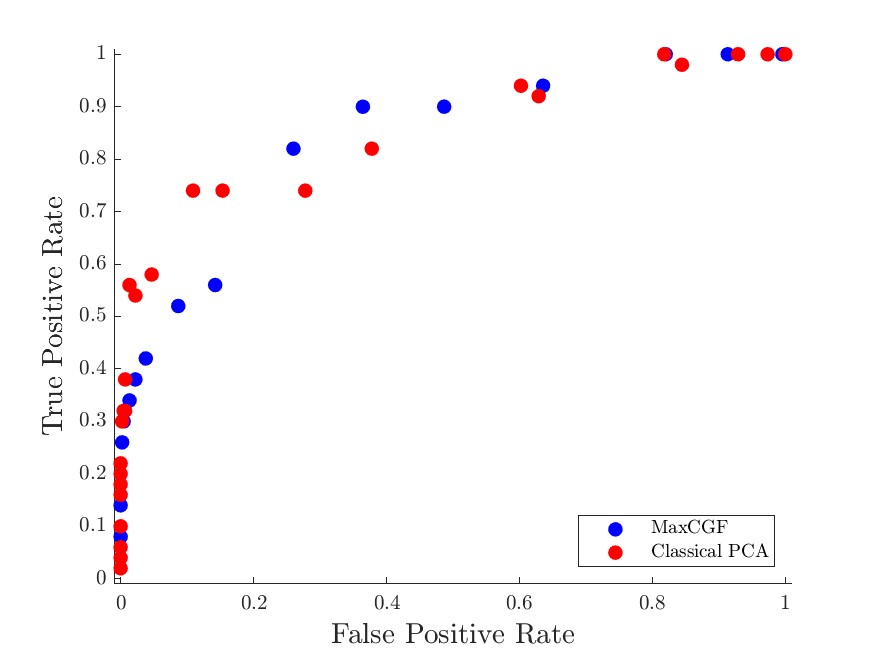}
		\label{fig:ROCStudtthetanu100}
	}
	\hfill
	\subfigure[$\nu=1000$]
	{
		\includegraphics[width=7cm]{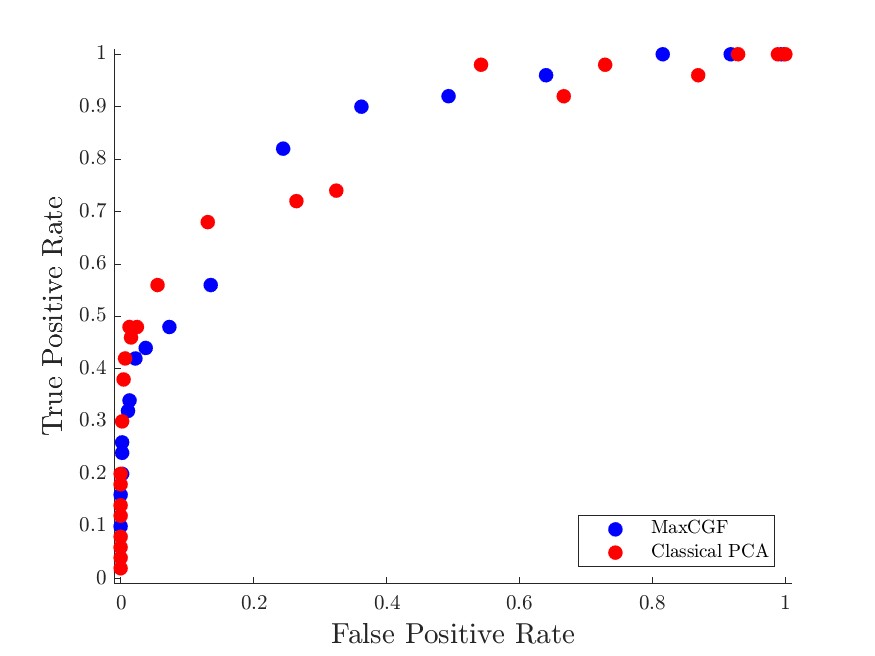}
		\label{fig:ROCStudtthetanu1000}
	}
	\caption{ROC curves using the CGF maximization and the classical PCA for the Student's t random vector.
}
	\label{fig:ROCStudttheta}
\end{figure}

\noindent
For these experiments, we also compare the performances of the Pe\~na-Prieto, Domino, and MaxCGF methods considering the degrees of freedom of the Student’s t random vector, $\nu=10$ and $\nu=30$. Note that, as mentioned at the beginning of Section \ref{sec:expana}, the Pe\~na-Prieto algorithm does not provide reasonable results.
\begin{table}[htbp]
					\centering
					\vspace{5pt}
					{\renewcommand{\arraystretch}{1.1}
						\hspace*{-10pt}{
\scalebox{0.9}{
\begin{tabular}{*{5}{!{\vrule width 0.9pt}p{3cm}}!{\vrule width 1.5pt}}
									\hline
									\textbf{Method} & \textbf{AUC} & \textbf{BCV} & \textbf{Time (min.)} & $\boldsymbol{\beta}^{*}$ \\ \hhline{|-|-|-|-|-|}
									\textbf{MaxCGF} & \cellcolor[rgb]{ 1,  0.44,  0.37} $0.8333$ & \cellcolor[rgb]{ 1,  0.44,  0.37}$0.5356$ & \cellcolor[rgb]{ 1,  0.44,  0.37}$10.0730$ & 7.25 \\
									\hhline{|-|-|-|-|-|}
									\textbf{Domino} & \cellcolor[rgb]{ 0,  .69,  .314} $0.8498$ & \cellcolor[rgb]{ 0,  .69,  .314}$0.5622$ & \cellcolor[rgb]{ 0,  .69,  .314}$8.6374$ & 7.50 \\
									\hhline{|-|-|-|-|-|}
									\textbf{Pe\~{n}a-Prieto} & - & - & - & - \\
									
									\hline
								\end{tabular}}}
							}
							\caption{Performances for the Student's t dataset for $\nu=10$.}
							\label{tab:Studtperfnu10}%
						\end{table}%
					\begin{table}[htbp]
						\centering
						\vspace{5pt}
						{\renewcommand{\arraystretch}{1.1}
							\hspace*{-10pt}{
\scalebox{0.9}{
\begin{tabular}{*{5}{!{\vrule width 0.9pt}p{3cm}}!{\vrule width 1.5pt}}
										\hline
										\textbf{Method} & \textbf{AUC} & \textbf{BCV} & \textbf{Time (min.)} & $\boldsymbol{\beta}^{*}$ \\ \hhline{|-|-|-|-|-|}
										\textbf{MaxCGF} & \cellcolor[rgb]{ 0,  .69,  .314}$0.9116$ & \cellcolor[rgb]{ 0,  .69,  .314}$0.7044$ & \cellcolor[rgb]{ 0,  .69,  .314}$8.0739$ & 7.75 \\
										\hhline{|-|-|-|-|-|}
										\textbf{Domino} & \cellcolor[rgb]{ 1,  0.44,  0.37} $0.8903$ & \cellcolor[rgb]{ 1,  0.44,  0.37}$0.6267$ & \cellcolor[rgb]{ 1,  0.44,  0.37}$19.1567$ & 7.25 \\
										\hhline{|-|-|-|-|-|}
										\textbf{Pe\~{n}a-Prieto} & - & - & - & - \\
										
										\hline
							\end{tabular}}}
						}
						\caption{Performances for the Student's t dataset for $\nu=30$.}
						\label{tab:Studtperfnu30}%
					\end{table}%
Tables \ref{tab:Studtperfnu10} and \ref{tab:Studtperfnu30} show the values of AUC, BCV, and CPU time for $\nu=10$ and $\nu=30$, respectively. In the former case, the Domino approach slightly outperforms the MaxCGF one, while in the latter case, the MaxCGF method yields better results in terms of AUC, BCV, and running time.
Figure \ref{fig:ROCStudt} exhibits the two methods' ROC curves, highlighting elevated outlier detection abilities in these experiments.
\begin{figure}[htbp]
			\centering
			\subfigure[$\nu=10$.]{\includegraphics[width=.45\textwidth]{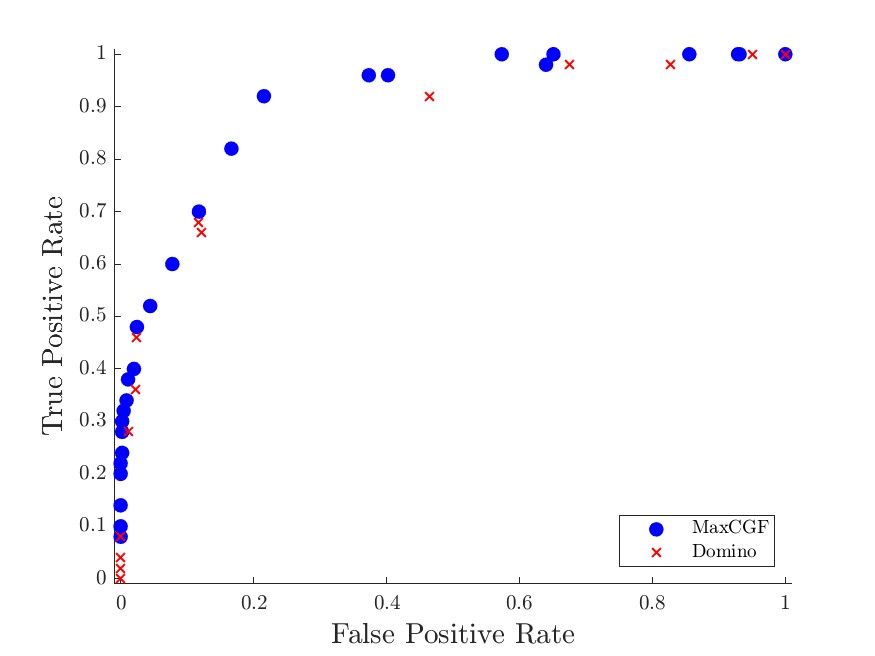}}
			\subfigure[$\nu=30$.]{\includegraphics[width=.45\textwidth]{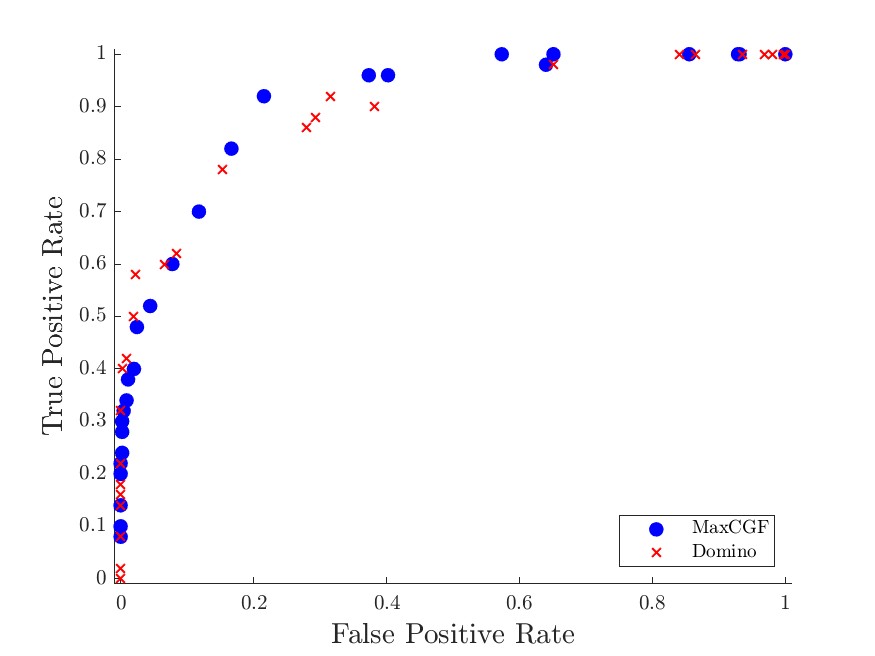}}
			\caption{ROC curves of the MaxCGF and Domino algorithms for the Student's t random vector.}
			\label{fig:ROCStudt}
		\end{figure}

\subsection{Computational results for financial real-world data}\label{subsec:realdata}

In this section, we use the three outlier detection methods to identify financial crises, which, as mentioned in the introduction, can be seen as periods when markets experience atypical behavior.
More precisely, we apply the outlier detection analysis to six real-world financial datasets
belonging to major stock markets across the world.
In Table \ref{tab:datasets}, we provide some details about these datasets, which consist of
daily prices, adjusted for dividends and stock splits, obtained from \emph{Refinitiv}.
\begin{table}[htbp]
 	\centering
 	\vspace{5pt}
 	{
 \scalebox{0.9}{
 \begin{tabular}{l|c|c|c|c}
 			Market   &	Abbrev.   & \# assets & Country & From-To \\    \hline
 			Dow Jones Industrial Average &	DJIA    &  30    & USA & \multirow{6}{*}{09/04/2019-23/03/2020}\\
Euro Stoxx 50 &	STOXX50 &  47    & Eurozone &  \\
DAX 30  	&	DAX     &  28    & Germany  &  \\
 CAC 40		&	CAC     &  40    &  France  &  \\
 FTSE 100	&	FTSE    &  98    &  UK      &  \\
EuroNext 100	&	N100    &  98    & Eurozone  & \\
 				\hline
 			\end{tabular}}}
 			\caption{List of the daily datasets analyzed.}
 			\label{tab:datasets}
 		\end{table}
From prices, we use both linear and logarithmic returns for the empirical analysis.
Since the results obtained are practically identical, we report here only those obtained by means of the linear returns.
The data analyzed are daily returns from April 9, 2019, to March 23, 2020 (approximately one financial year).
We have chosen this time horizon because it contains a period of high instability due to the COVID-19 pandemic, the effects of which occur approximately in early February 2020 (see, e.g. \cite{shu2021covid}).
%
\begin{table}[htbp]
	\centering
	\vspace{5pt}
	{\renewcommand{\arraystretch}{1.1}
		\hspace*{-10pt}{\small{
	\begin{tabular}{*{5}{!{\vrule width 0.9pt}p{3cm}}!{\vrule width 1.5pt}}
		\hline
		\textbf{Method} & \textbf{AUC} & \textbf{BCV} & \textbf{Time (min.)} & $\boldsymbol{\beta}^{*}$ \\ \hhline{|-|-|-|-|-|}
		\textbf{MaxCGF} & \cellcolor[rgb]{ 0,  .69,  .314}$0.9057$ & \cellcolor[rgb]{ 0,  .69,  .314}$0.7156$ & \cellcolor[rgb]{ 0,  .69,  .314}$5.9174$ & 6.75 \\ \hhline{|-|-|-|-|-|}
		\textbf{Domino} & \cellcolor[rgb]{ 1,  1,  0}$0.8282$ & \cellcolor[rgb]{ 1,  1,  0}$0.6178$ & \cellcolor[rgb]{ 1,  1,  0}$25.6869$ & 5.75 \\
		\hhline{|-|-|-|-|-|}
		\textbf{Pe\~{n}a-Prieto} & \cellcolor[rgb]{ 1,  0.44,  0.37}0.7873 & \cellcolor[rgb]{ 1,  0.44,  0.37}$0.5511$ & \cellcolor[rgb]{ 1,  0.44,  0.37}$84.0113$ & 6.50 \\
		\hline
	\end{tabular}}}
}
	\caption{Performances for DJIA.}
	\label{tab:DJIAperf}%
\end{table}%
\begin{table}[htbp]
	\centering
	\vspace{5pt}
	{\renewcommand{\arraystretch}{1.1}
		\hspace*{-10pt}{\small{
				\begin{tabular}{*{5}{!{\vrule width 0.9pt}p{3cm}}!{\vrule width 1.5pt}}
					\hline
					\textbf{Method} & \textbf{AUC} & \textbf{BCV} & \textbf{Time (min.)} & $\boldsymbol{\beta}^{*}$ \\ \hhline{|-|-|-|-|-|}
					\textbf{MaxCGF} & \cellcolor[rgb]{ 0,  .69,  .314} $0.8312$ & \cellcolor[rgb]{ 0,  .69,  .314}$0.5666$ & \cellcolor[rgb]{ 0,  .69,  .314}$3.6558$ & 8.00 \\ \hhline{|-|-|-|-|-|}
					\textbf{Domino} & \cellcolor[rgb]{ 1,  0.44,  0.37} $0.7480$ & \cellcolor[rgb]{ 1,  0.44,  0.37}$0.4713$ & \cellcolor[rgb]{ 1,  0.44,  0.37}$92.4653$ & 7.50 \\ \hhline{|-|-|-|-|-|}
					\textbf{Pe\~{n}a-Prieto} & - & - & - & - \\
				\hline
			\end{tabular}}}
		}
			\caption{Performances for STOXX50.}
			\label{tab:ESTXperf}%
		\end{table}%
		\begin{table}[htbp]
			\centering
			\vspace{5pt}
			{\renewcommand{\arraystretch}{1.1}
				\hspace*{-10pt}{\small{
						\begin{tabular}{*{5}{!{\vrule width 0.9pt}p{3cm}}!{\vrule width 1.5pt}}
							\hline
							\textbf{Method} & \textbf{AUC} & \textbf{BCV} & \textbf{Time (min.)} & $\boldsymbol{\beta}^{*}$ \\ \hhline{|-|-|-|-|-|}
							\textbf{MaxCGF} & \cellcolor[rgb]{ 0,  .69,  .314}$0.8866$ & \cellcolor[rgb]{ 0,  .69,  .314}$0.6555$ & \cellcolor[rgb]{ 0,  .69,  .314}$3.8149$ & 5.75 \\ \hhline{|-|-|-|-|-|}
							\textbf{Domino} & \cellcolor[rgb]{ 1,  0.44,  0.37}$0.7924$ & \cellcolor[rgb]{ 1,  0.44,  0.37}$0.5701$ & \cellcolor[rgb]{ 1,  0.44,  0.37}$17.2033$ & 7.00\\ \hhline{|-|-|-|-|-|}
							\textbf{Pe\~{n}a-Prieto} & - &  - & - & -\\
							
							\hline
						\end{tabular}}}
					}
					\caption{Performances for DAX.}
					\label{tab:DAXperf}%
				\end{table}%
\begin{table}[htbp]
	\centering
	\vspace{5pt}
	{\renewcommand{\arraystretch}{1.1}
		\hspace*{-10pt}{\small{
				\begin{tabular}{*{5}{!{\vrule width 0.9pt}p{3cm}}!{\vrule width 1.5pt}}
					\hline
					\textbf{Method} & \textbf{AUC} & \textbf{BCV} & \textbf{Time (min.)} & $\boldsymbol{\beta}^{*}$ \\ \hhline{|-|-|-|-|-|}
					\textbf{MaxCGF} & \cellcolor[rgb]{ 0,  .69,  .314} $0.8326$ & \cellcolor[rgb]{ 0,  .69,  .314}$0.5187$ & \cellcolor[rgb]{ 0,  .69,  .314}$3.7390$ & 8.25 \\
					 \hhline{|-|-|-|-|-|}
					\textbf{Domino} & \cellcolor[rgb]{ 1,  0.44,  0.37} $0.7391$ & \cellcolor[rgb]{ 1,  0.44,  0.37}$0.4637$ & \cellcolor[rgb]{ 1,  0.44,  0.37}$75.1684$ & 6.50 \\ \hhline{|-|-|-|-|-|}
					\textbf{Pe\~{n}a-Prieto} &  - & - & - & - \\					
					\hline
				\end{tabular}}}
			}
			\caption{Performances for CAC.}
			\label{tab:CAC40perf}%
		\end{table}%
\begin{table}[htbp]
	\centering
	\vspace{5pt}
	{\renewcommand{\arraystretch}{1.1}
		\hspace*{-10pt}{\small{
				\begin{tabular}{*{5}{!{\vrule width 0.9pt}p{3cm}}!{\vrule width 1.5pt}}
					\hline
					\textbf{Method} & \textbf{AUC} & \textbf{BCV} & \textbf{Time (min.)} & $\boldsymbol{\beta}^{*}$ \\ \hhline{|-|-|-|-|-|}
					\textbf{MaxCGF} & \cellcolor[rgb]{ 0,  .69,  .314} $0.8021$ & \cellcolor[rgb]{ 0,  .69,  .314}$0.5207$ & \cellcolor[rgb]{ 0,  .69,  .314}$5.4309$ & 7.50 \\					
					\hhline{|-|-|-|-|-|}
					\textbf{Domino} & \cellcolor[rgb]{ 1,  0.44,  0.37}$0.7061$ & \cellcolor[rgb]{ 1,  0.44,  0.37}$0.4641$ & \cellcolor[rgb]{ 1,  0.44,  0.37}$2644.9669$ & 7.50 \\ \hhline{|-|-|-|-|-|}
					\textbf{Pe\~{n}a-Prieto} & - & - & - & - \\				
					\hline
				\end{tabular}}}
			}
			\caption{Performances for FTSE.}
			\label{tab:FTSEperf}%
		\end{table}%
		\begin{table}[htbp]
			\centering
			\vspace{5pt}
			{\renewcommand{\arraystretch}{1.1}
				\hspace*{-10pt}{\small{
						\begin{tabular}{*{5}{!{\vrule width 0.9pt}p{3cm}}!{\vrule width 1.5pt}}
							\hline
							\textbf{Method} & \textbf{AUC} & \textbf{BCV} & \textbf{Time (min.)} & $\boldsymbol{\beta}^{*}$ \\ \hhline{|-|-|-|-|-|}
							\textbf{MaxCGF} & \cellcolor[rgb]{ 0,  .69,  .314}$0.8643$ & \cellcolor[rgb]{ 0,  .69,  .314}$0.5474$ & \cellcolor[rgb]{ 0,  .69,  .314}$6.2306$ & 7.75 \\ \hhline{|-|-|-|-|-|}
							\textbf{Domino} & \cellcolor[rgb]{ 1,  0.44,  0.37}$0.6794$ & \cellcolor[rgb]{ 1,  0.44,  0.37}$0.3624$ & \cellcolor[rgb]{ 1,  0.44,  0.37}$2986.0548$ & 7.75 \\
							\hhline{|-|-|-|-|-|}
							\textbf{Pe\~{n}a-Prieto} & - & - & - & - \\
							\hline
						\end{tabular}}}
					}
					\caption{Performances for N100.}
					\label{tab:ENXperf}%
				\end{table}%

\noindent
In Tables \ref{tab:DJIAperf}-\ref{tab:ENXperf}, we report the computational results for all the datasets listed in Table \ref{tab:datasets}.
These tables show that AUC obtained by the MaxCGF approach is between 0.8021 and 0.9057, and is always higher than that achieved by the other two methods.
\begin{figure}[htbp]
	\centering
	\subfigure[Value of the Dow Jones index. The red line indicates the beginning of the crisis.]{\includegraphics[width=.45\textwidth]{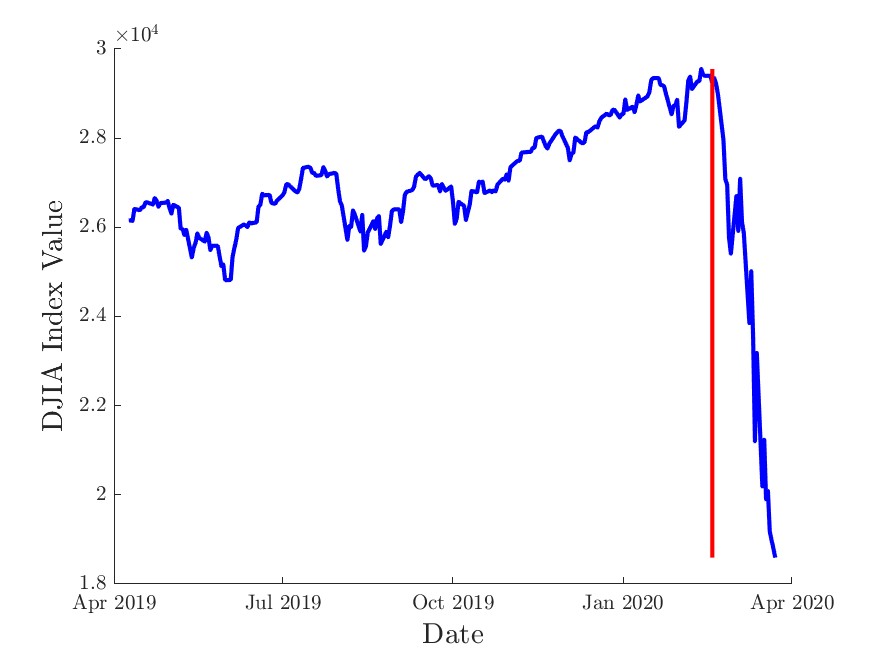}}
	\subfigure[ROC curves of the three algorithms for the DJIA dataset.]{\includegraphics[width=.45\textwidth]{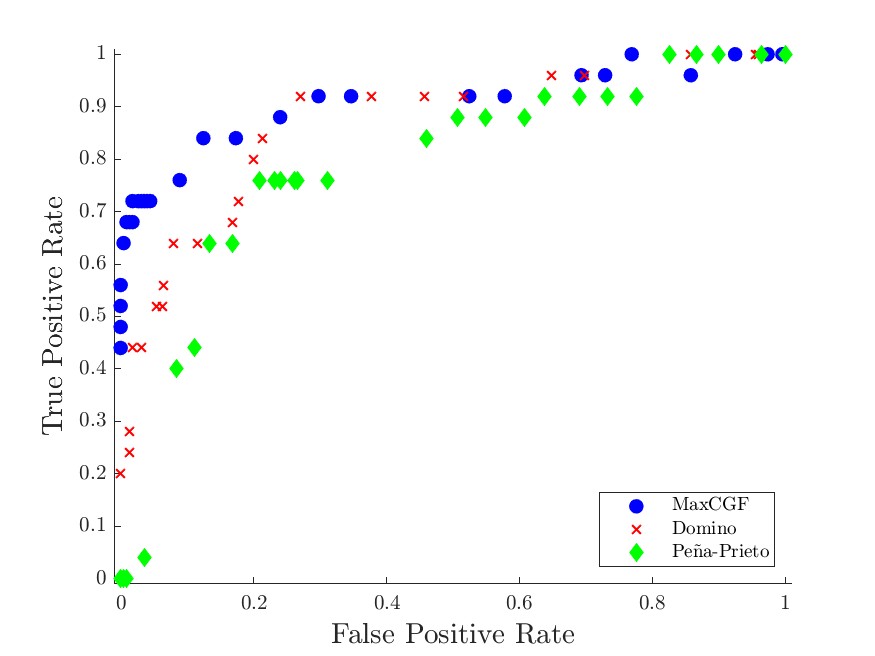}}
	\caption{DJIA dataset.}
	\label{fig:DJIA}
\end{figure}
\begin{figure}[htbp]
	\centering
	\subfigure[Value of the Eurostoxx50 index. The red line indicates the beginning of the crisis.]{\includegraphics[width=.45\textwidth]{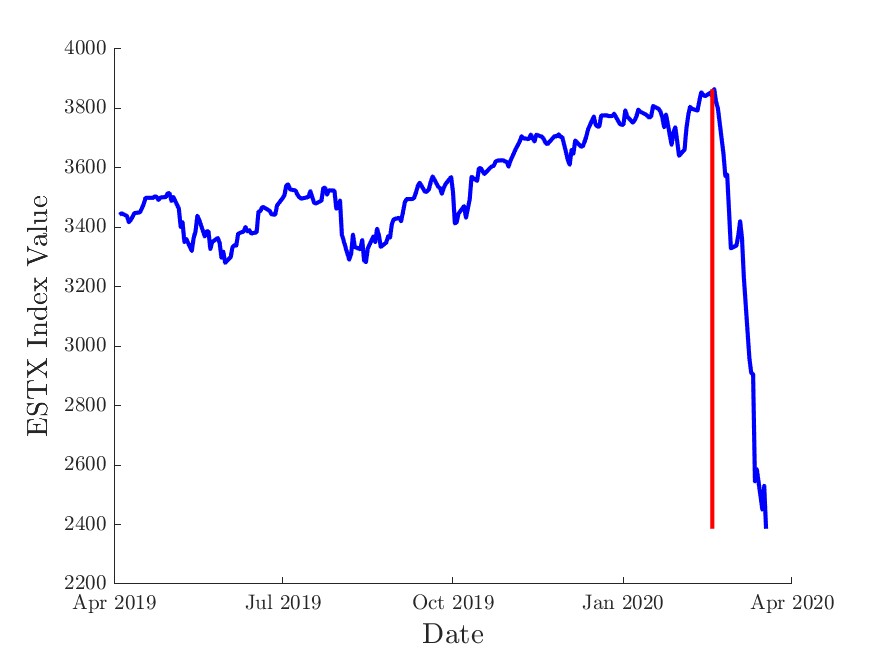}}
\subfigure[ROC curves of the MaxCGF and Domino algorithms for the STOXX50 dataset.]{\includegraphics[width=.45\textwidth]{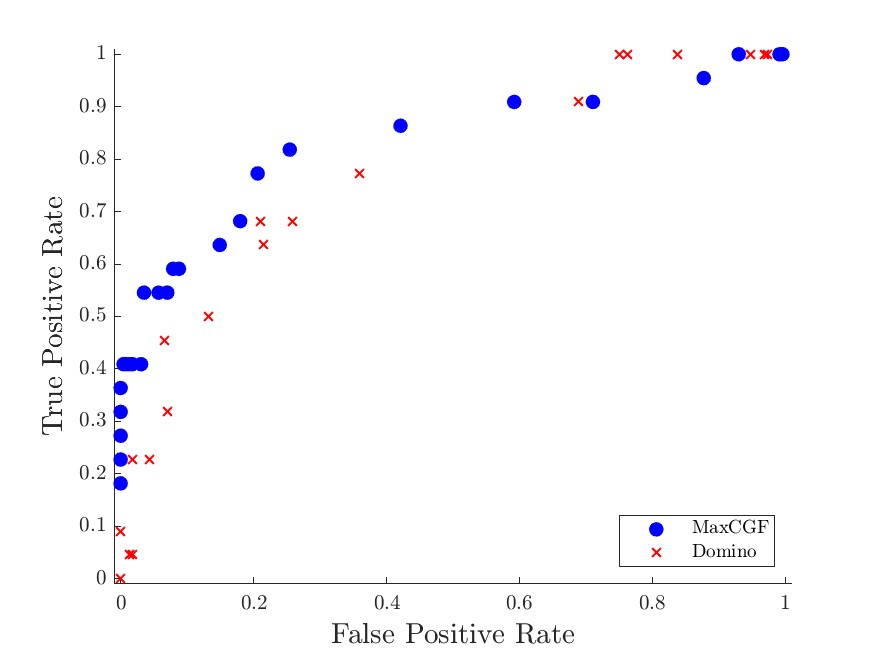}}
	\caption{STOXX50 dataset.}
	\label{fig:ESTX}
\end{figure}
\begin{figure}[htbp]
	\centering
	\subfigure[Value of the DAX index. The red line indicates the beginning of the crisis.]{\includegraphics[width=.45\textwidth]{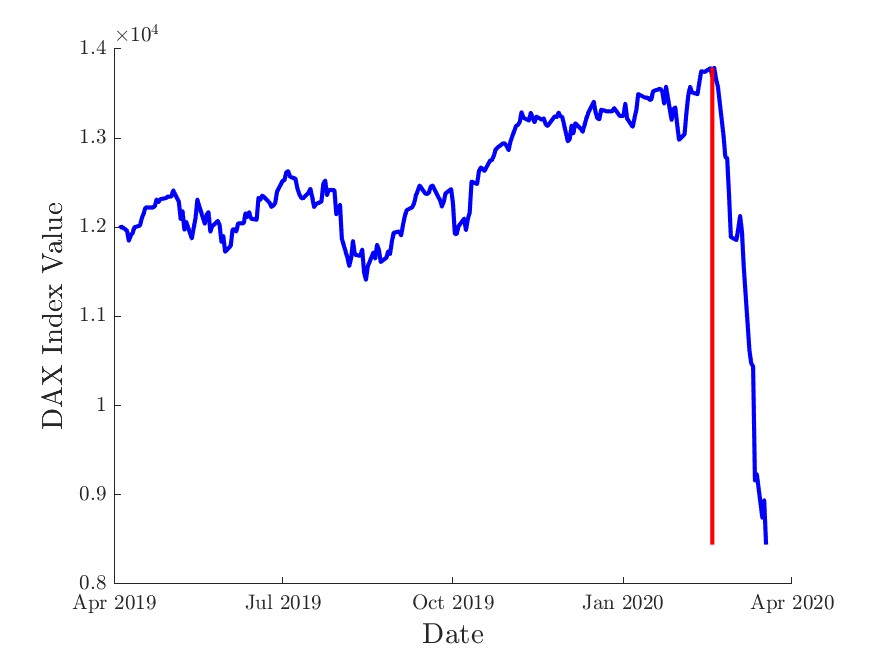}}
\subfigure[ROC curves of the MaxCGF and Domino algorithms for the DAX dataset.]{\includegraphics[width=.45\textwidth]{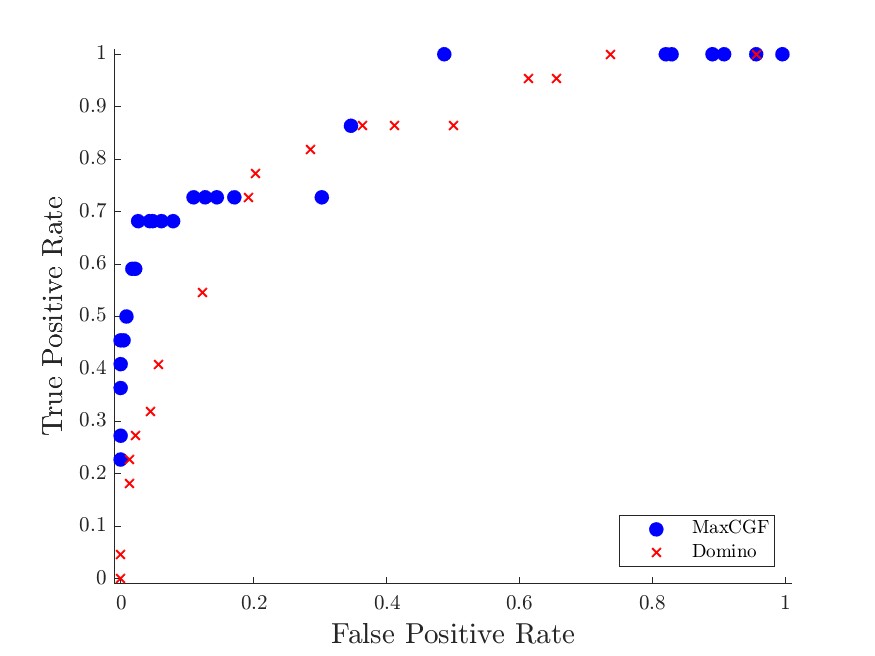}}
	\caption{DAX dataset.}
	\label{fig:DAX}
\end{figure}
\begin{figure}[htbp]
	\centering
	\subfigure[Value of the CAC40 index. The red line indicates the beginning of the crisis.]{\includegraphics[width=.45\textwidth]{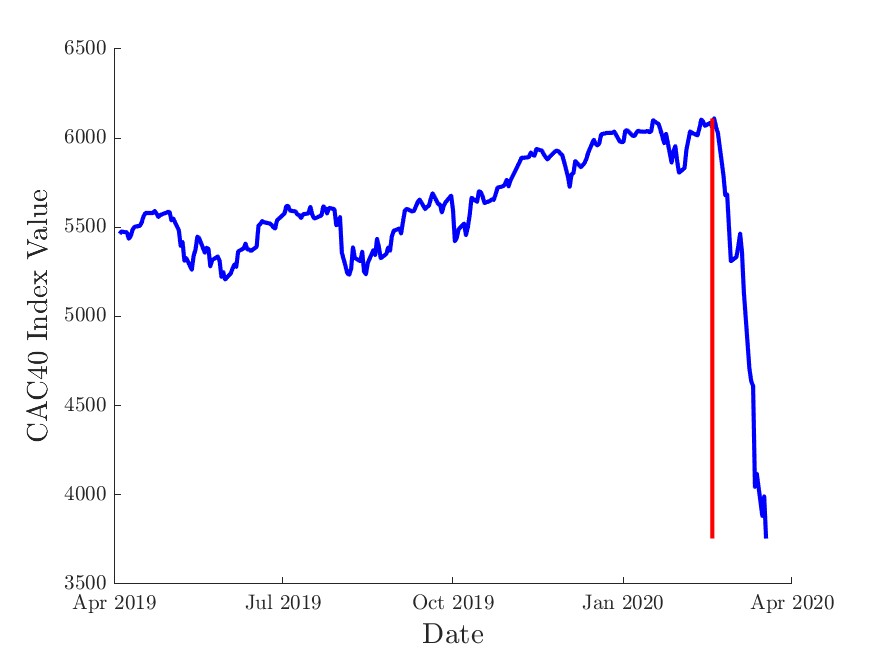}}
\subfigure[ROC curves of the MaxCGF and Domino algorithms for the CAC40 dataset.]{\includegraphics[width=.45\textwidth]{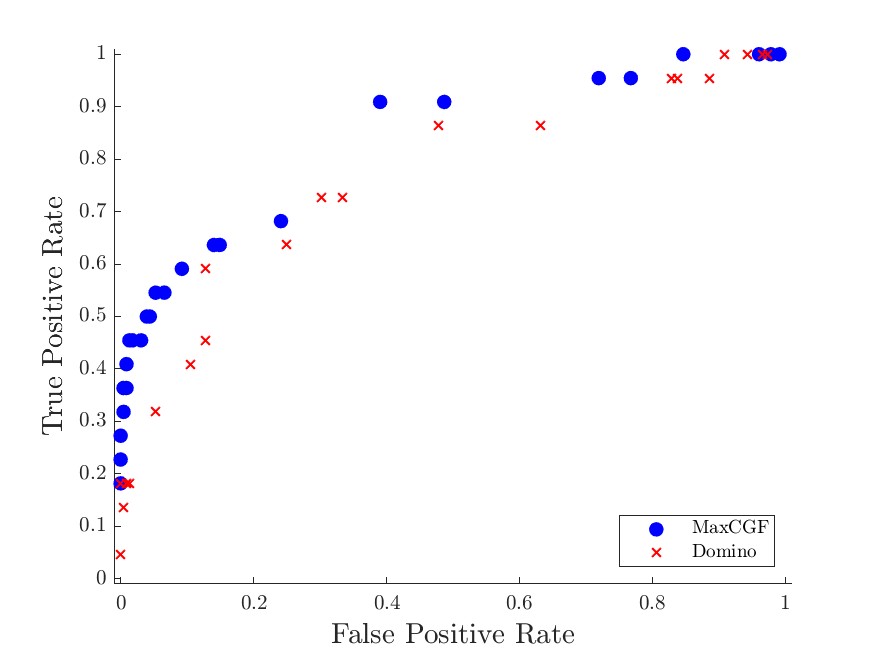}}
	\caption{CAC40 dataset.}
	\label{fig:CAC40}
\end{figure}
\begin{figure}[htbp]
	\centering
	\subfigure[Value of the FTSE100 index. The red line indicates the beginning of the crisis.]{\includegraphics[width=.45\textwidth]{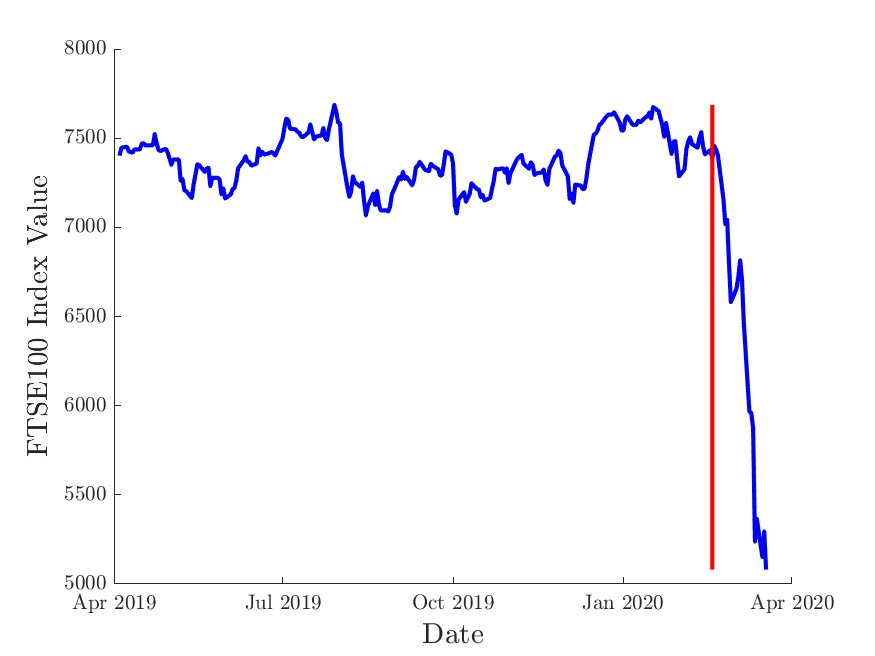}}
\subfigure[ROC curves of the MaxCGF and Domino algorithms for the FTSE100 dataset.]{\includegraphics[width=.45\textwidth]{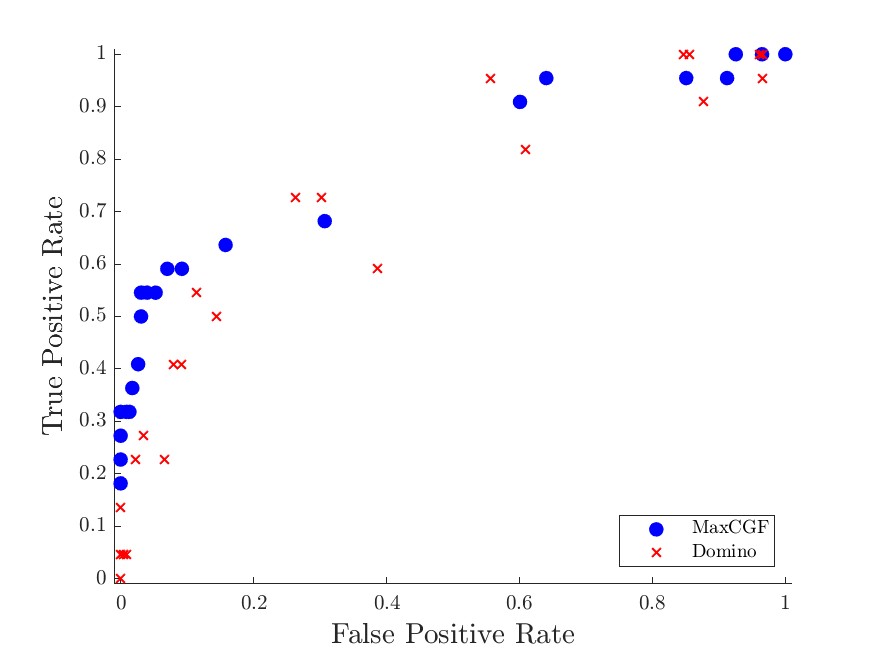}}
	\caption{FTSE100 dataset.}
	\label{fig:FTSE100}
\end{figure}
\begin{figure}[htbp]
	\centering
	\subfigure[Value of the Euro Next 100 index. The red line indicates the beginning of the crisis.]{\includegraphics[width=.45\textwidth]{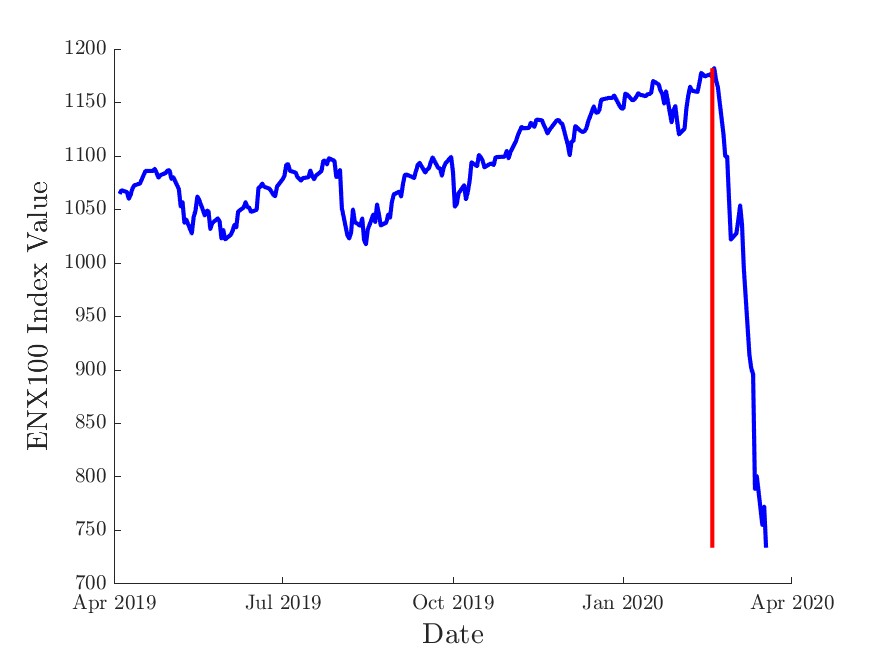}}
\subfigure[ROC curves of the MaxCGF and Domino algorithms for the N100 dataset.]{\includegraphics[width=.45\textwidth]{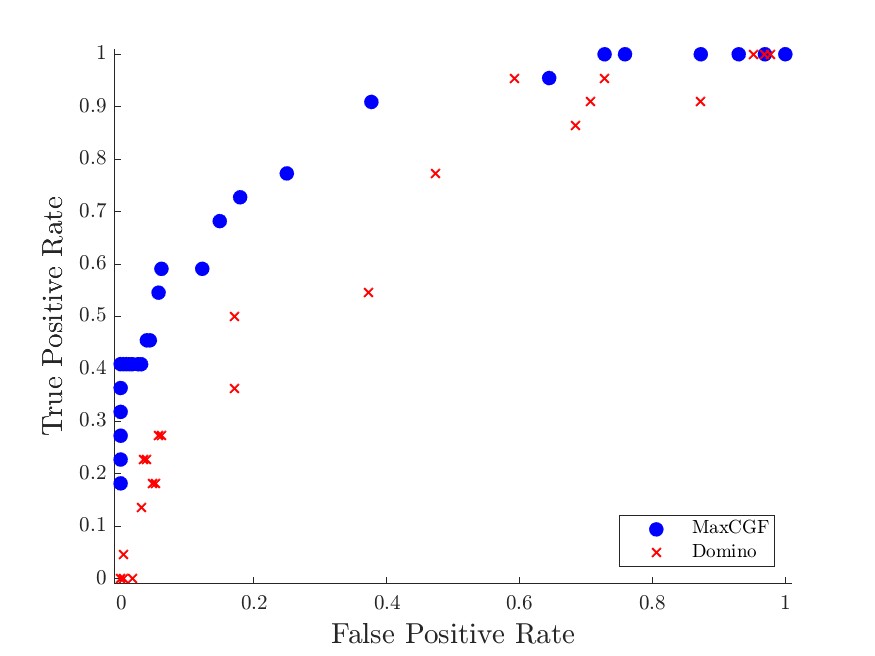}}
	\caption{N100 dataset.}
	\label{fig:ENX}
\end{figure}

\noindent
In Figures \ref{fig:DJIA}-\ref{fig:ENX}, we also provide the ROC curves for all the financial market analyzed.
Interestingly, TPR of the MaxCGF algorithm is high for values of FPR slightly higher than 0.1.
This means that our algorithm is able to correctly detect a large portion of outliers misclassifying only few ordinary data (as outlier) even for high values of the threshold $\beta$.
Furthermore, the MaxCGF method is always the least time-consuming, only taking a few minutes to complete the analysis for all the financial datasets analyzed.
Conversely, both the Domino and the Pe\~na-Prieto algorithms are heavily influenced by the number of assets $n$.
Indeed, on the one hand, the Pe\~na-Prieto procedure is able to produce results only for DJIA. On the other hand, for the Domino approach, the running time ranges from about 17 minutes for DAX ($n=28$) to around two days for the N100 ($n=98$).
Summing up, as highlighted in the tables, the MaxCGF method always shows better results for all the examined performance measures, followed by the Domino method.

\section{Conclusions}\label{sec:conc}

In this paper, we have proposed a non-parametric approach to detect anomalies in data by examining their univariate projections on appropriate directions that depend on cumulants of any order.
Such directions, where the original data have been projected,
are those that maximize the cumulant generating function (CGF).

\noindent
In this respect, we have first refined some theoretical results of \cite{bernacchia2008detecting1} and \cite{bernacchia2008detecting2} investigating
the directions that maximize CGF of data with normal and skew-normal distributions.
Then, we have proved in the general non-parametric case that CGF is a convex function and characterized the CGF maximization problem on the unit $n$-circle as a concave minimization problem.
Furthermore, we have extended the outlier detection methodology based on the projections of multivariate data on the directions obtained by the classical PCA technique
to the directions that maximize CGF.
Finally, we have presented an extensive empirical analysis testing the performance of our outlier detection procedure, named MaxCGF, and comparing it with two other methods, proposed by \cite{domino2020multivariate} and \cite{pena2001multivariate}.
From the computational results, we have observed that the Pe\~na-Prieto method shows high performance for standard normal and normal simulated data, while, in other cases, its ability to detect outliers seems to be low.
The Domino method typically provides intermediate performance, but favorable in the case of Student's t distributed data.
The MaxCGF approach also performs well for normal and Student's random vectors.
However, for skew-normal data or real-world financial data, the MaxCGF method always shows better results for all performance measures examined, including efficiency.

\begin{appendices}
	\section{Relative variance of the CGF estimator}
	\label{RelErr}

The sample estimate of the Cumulant Generating
Function (CGF) of a discrete multivariate variable $\boldsymbol{X}_{t}=(X_{1,t},\ldots, X_{n,t})$ with $t=1,\ldots,T$ is
\begin{equation*}
	G_{\boldsymbol{X}}(\boldsymbol{\xi})=\ln \displaystyle\frac{1}{T} \sum_{t=1}^{T} e^{\boldsymbol{\xi}^{T} \boldsymbol{X}_{t}} = \ln M_{\boldsymbol{X}}(\boldsymbol{\xi}) \, ,
\end{equation*}
where, therefore, $M_{\boldsymbol{X}}(\boldsymbol{\xi})$ denotes the sample estimate of the moment generating function.
The relative variance of the CGF estimator is then defined as
\begin{equation}\label{eq:RelVarCGF_}
  \varepsilon_{G}^2 = \frac{\V[G]}{(\E[G])^2} \, .
\end{equation}
To find an explicit expression for \eqref{eq:RelVarCGF_}, we exploit the Taylor expansion for moments of functions of random variables \citep[see, e.g.,][]{benaroya2005probability}.
More precisely, we use the Taylor expansion around $\mu_M = \E[M]$, namely
\begin{equation}\label{eq:CGF_Taylor}
G = \ln M = \ln \mu_M + \frac{1}{\mu_M} (M - \mu_M) - \frac{1}{2} \frac{1}{\mu_{M}^{2}} (M - \mu_M)^2 + \ldots
\end{equation}
%
%
Using the first order approximation, we obtain
\begin{equation}\label{eq:Exp_CGF_Taylor_First}
\E[G] = \E[\ln M] \simeq \E[ \ln \mu_M + \frac{1}{\mu_M} (M - \mu_M) ] = \ln \mu_M
\end{equation}
%
\begin{equation}\label{eq:Var_CGF_Taylor_First}
\V[G] = \V[\ln M] \simeq \V[ \ln \mu_M + \frac{1}{\mu_M} (M - \mu_M) ] = \frac{1}{\mu_{M}^{2}} \sigma_{M}^{2} \, .
\end{equation}
Thus,
\begin{equation}\label{eq:RelVarCGF_Taylor_First_}
  \varepsilon_{G}^{2} \simeq \frac{\sigma_{M}^{2}}{\mu_{M}^{2} (\ln \mu_M)^2} \, ,
\end{equation}
where $\mu_M = \E[M]$ and $\sigma_{M}^{2} = \V[M]$.

\noindent	
The sample estimate of the moment generating function of a discrete multivariate variable $\boldsymbol{X}_{t}=(X_{1,t},\ldots, X_{n,t})$ with $t=1,\ldots,T$ is
\begin{equation}\label{eq:MGFApp}
  M_{\boldsymbol{X}}(\boldsymbol{\xi})= \displaystyle\frac{1}{T} \sum_{t=1}^{T} e^{\boldsymbol{\xi}^{T} \boldsymbol{X}_{t}} \, ,
\end{equation}
where, as mentioned in Section \ref{sec:CGFMaximum}, to explicitly find \eqref{eq:RelVarCGF_Taylor_First_},
we assume i.i.d. normally distributed random vectors $\boldsymbol{X}_{t}\sim N_{n}(\boldsymbol{0}, \boldsymbol{\Sigma})$ $\forall t =1, \ldots, T$.
Then, the expectation of the estimator \eqref{eq:MGFApp} is
\begin{equation}\label{eq:MGFApp2}
  \E[M]= \displaystyle\frac{1}{T} \sum_{t=1}^{T} \E[e^{\boldsymbol{\xi}^{T} \boldsymbol{X}_{t}}] \, ,
\end{equation}
where $\E[e^{\boldsymbol{\xi}^{T} \boldsymbol{X}_{t}}]=\int\cdots\int e^{\boldsymbol{\xi}\boldsymbol{x}} f_{\boldsymbol{X}_{t}}(\boldsymbol{x})d\boldsymbol{x}$ and $f_{\boldsymbol{X}_{t}}(\boldsymbol{x})=2\pi^{-\frac{n}{2}} \det(\boldsymbol{\Sigma})^{-\frac{1}{2}} e^{-\frac{1}{2}\boldsymbol{x}^{T} \boldsymbol{\Sigma}^{-1}\boldsymbol{x}}$.
Hence,
\begin{eqnarray}
  \E[e^{\boldsymbol{\xi}^{T} \boldsymbol{X}_{t}}] &=& 2 \pi^{-\frac{n}{2}} \det(\boldsymbol{\Sigma})^{-\frac{1}{2}} \int \cdots \int e^{\boldsymbol{\xi}\boldsymbol{x}} e^{-\frac{1}{2} \boldsymbol{x}^{T} \boldsymbol{\Sigma}^{-1} \boldsymbol{x}}d\boldsymbol{x} \nonumber \\
   &=& e^{\frac{1}{2} \boldsymbol{\xi}^{T} \boldsymbol{\Sigma} \boldsymbol{\xi}} \, 2 \pi^{-\frac{n}{2}} \det(\boldsymbol{\Sigma})^{-\frac{1}{2}} \int \cdots \int e^{-\frac{1}{2} (\boldsymbol{x} - \boldsymbol{\Sigma} \boldsymbol{\xi})^{T} \boldsymbol{\Sigma}^{-1} (\boldsymbol{x}-\boldsymbol{\Sigma} \boldsymbol{\xi})} d\boldsymbol{x} \nonumber \\
   &=& e^{\frac{1}{2} \boldsymbol{\xi}^{T} \boldsymbol{\Sigma} \boldsymbol{\xi}} \, , \label{eq:MGFApp3}
\end{eqnarray}
since $2 \pi^{-\frac{n}{2}} \det(\boldsymbol{\Sigma})^{-\frac{1}{2}} \int \cdots \int e^{-\frac{1}{2} (\boldsymbol{x} - \boldsymbol{\Sigma} \boldsymbol{\xi})^{T} \boldsymbol{\Sigma}^{-1} (\boldsymbol{x}-\boldsymbol{\Sigma} \boldsymbol{\xi})} d\boldsymbol{x}=1$.
Thus, using \eqref{eq:MGFApp3} in \eqref{eq:MGFApp2}, we obtain
\begin{equation}\label{eq:MGFApp4}
  \mu_{M} = \E[M]= \displaystyle\frac{1}{T} \, T \, e^{\frac{1}{2} \boldsymbol{\xi}^{T} \boldsymbol{\Sigma} \boldsymbol{\xi}} = e^{\frac{1}{2} \boldsymbol{\xi}^{T} \boldsymbol{\Sigma} \boldsymbol{\xi}}
\end{equation}
For the variance of the estimator \eqref{eq:MGFApp}, we have
\begin{equation}\label{eq:MGFApp5}
  \V[M] =\E[M^2] - \E[M]^{2} \, ,
\end{equation}
where
\begin{eqnarray*}
  \E[M^{2}] &=& \E \biggl[ \frac{1}{T} \sum_{t=1}^{T} e^{\boldsymbol{\xi}^{T} \boldsymbol{X}_{t} } \, \frac{1}{T} \sum_{j=1}^{T} e^{\boldsymbol{\xi}^{T} \boldsymbol{X}_{j} } \biggr] \\
   &=& \displaystyle \frac{1}{T^{2}} \E\biggl[ \sum_{t=1}^{T} e^{ 2 \boldsymbol{\xi}^{T} \boldsymbol{X}_{t}} + \sum_{				
			t \neq j
		}^{T}
		( e^{\boldsymbol{\xi}^{T} \boldsymbol{X}_{t}} )
		(e^{\boldsymbol{\xi}^{T} \boldsymbol{X}_{j}})
		\biggr] \\
   &=&  \displaystyle \frac{1}{T^{2}} \sum_{t=1}^{T} \E\biggl[ e^{ 2 \boldsymbol{\xi}^{T} \boldsymbol{X}_{t}} \biggr] + \sum_{				
			t \neq j
		}^{T}
		\E\biggl[ ( e^{\boldsymbol{\xi}^{T} \boldsymbol{X}_{t}} )
		(e^{\boldsymbol{\xi}^{T} \boldsymbol{X}_{j}})
		\biggr] \, .
\end{eqnarray*}
Now, similarly to \eqref{eq:MGFApp3}, we obtain that
\begin{equation}
	\E [e^{2\boldsymbol{\xi}\boldsymbol{X}_{t}}] = e^{2 \boldsymbol{\xi}^{T} \boldsymbol{\Sigma} \boldsymbol{\xi}} \, .
	\label{eq:mean1}
	\end{equation}
Furthermore, since, by assumption, $\boldsymbol{X}_{t}$ are i.i.d. $\forall t =1, \ldots, T$, we have
\begin{eqnarray}
  \E[M^{2}] &=& \displaystyle \frac{1}{T^{2}} (T \E [ e^{2\boldsymbol{\xi}^{T} \boldsymbol{X}}] +(T^2-T) (\E [e^{\boldsymbol{\xi}^{T} \boldsymbol{X}}])^{2} \nonumber \\
   &=& \frac{1}{T} \E [e^{2 \boldsymbol{\xi}^{T} \boldsymbol{X}}] + (1-\frac{1}{T}) (\E [e^{\boldsymbol{\xi}^{T}\boldsymbol{X}}])^{2} \nonumber \\
   &=& \frac{1}{T} e^{2 \boldsymbol{\xi}^{T} \boldsymbol{\Sigma} \boldsymbol{\xi}} + (1-\frac{1}{T}) e^{\boldsymbol{\xi}^{T} \boldsymbol{\Sigma} \boldsymbol{\xi}} \, . \label{eq:Var1a}
\end{eqnarray}
Substituting \eqref{eq:Var1a} in Expression \eqref{eq:MGFApp5}, we can write
\begin{eqnarray}
  \sigma_{M}^{2} = \V[M]   &=& \frac{1}{T} e^{2 \boldsymbol{\xi}^{T} \boldsymbol{\Sigma} \boldsymbol{\xi}} + (1-\frac{1}{T}) e^{\boldsymbol{\xi}^{T} \boldsymbol{\Sigma} \boldsymbol{\xi}} - e^{\boldsymbol{\xi}^{T} \boldsymbol{\Sigma} \boldsymbol{\xi}} \nonumber \\
   &=& \frac{1}{T} \left( e^{2 \boldsymbol{\xi}^{T} \boldsymbol{\Sigma} \boldsymbol{\xi}} - e^{\boldsymbol{\xi}^{T} \boldsymbol{\Sigma} \boldsymbol{\xi}} \right)  \label{eq:MGFApp4bb}
\end{eqnarray}
Thus, using \eqref{eq:MGFApp4} and \eqref{eq:MGFApp4bb} in \eqref{eq:RelVarCGF_},
\begin{equation}\label{eq:RelVarCGF_Taylor_First}
  \varepsilon_{G}^2 \simeq \frac{\sigma_{M}^{2}}{\mu_{M}^{2} (\ln \mu_M)^2} =
  \frac{\frac{1}{T} \left( e^{2 \boldsymbol{\xi}^{T} \boldsymbol{\Sigma} \boldsymbol{\xi}} - e^{\boldsymbol{\xi}^{T} \boldsymbol{\Sigma} \boldsymbol{\xi}} \right)}{e^{ \boldsymbol{\xi}^{T} \boldsymbol{\Sigma} \boldsymbol{\xi}} ( \frac{1}{2} \boldsymbol{\xi}^{T} \boldsymbol{\Sigma} \boldsymbol{\xi} )^2 } = \frac{1}{T} \frac{e^{ \boldsymbol{\xi}^{T} \boldsymbol{\Sigma} \boldsymbol{\xi}} -1}{( \frac{1}{2} \boldsymbol{\xi}^{T} \boldsymbol{\Sigma} \boldsymbol{\xi} )^2} \, .
\end{equation}
Denoting $\boldsymbol{\xi}=r \boldsymbol{\theta}$,
we obtain
\begin{equation}\label{eq:RelVarCGF_Taylor_First2}
  \varepsilon_{G}^2 \simeq \frac{4}{T} \frac{e^{ r^2 \boldsymbol{\theta}^{T} \boldsymbol{\Sigma} \boldsymbol{\theta}} -1}{r^4 ( \boldsymbol{\theta}^{T} \boldsymbol{\Sigma} \boldsymbol{\theta} )^2} \, .
\end{equation}
Simplifying the issue, we substitute $\boldsymbol{\theta}^{T} \boldsymbol{\Sigma} \boldsymbol{\theta}$ with its largest eigenvalue $\lambda_1$ (computed by the standard PCA), and, therefore, Expression \eqref{eq:RelVarCGF_Taylor_First2} becomes
\begin{equation}\label{eq:RelVarCGF_Taylor_First3}
  \varepsilon_{G}^2 \simeq \frac{4}{T} \frac{e^{ r^2 \lambda_1 } -1}{r^4 ( \lambda_1 )^2} \, .
\end{equation}

\end{appendices}

{\footnotesize
\bibliographystyle{apa}

\begin{thebibliography}{}

\bibitem[\protect\astroncite{An{\'e} et~al.}{2008}]{ane2008robust}
An{\'e}, T., Ureche-Rangau, L., Gambet, J.-B., and Bouverot, J. (2008).
\newblock Robust outlier detection for asia--pacific stock index returns.
\newblock {\em Journal of International Financial Markets, Institutions and
  Money}, 18(4):326--343.

\bibitem[\protect\astroncite{Arellano-Valle and
  Azzalini}{2008}]{arellano2008centred}
Arellano-Valle, R.~B. and Azzalini, A. (2008).
\newblock The centred parametrization for the multivariate skew-normal
  distribution.
\newblock {\em Journal of Multivariate Analysis}, 99(7):1362--1382.

\bibitem[\protect\astroncite{Azzalini and
  Capitanio}{1999}]{azzalini1999statistical}
Azzalini, A. and Capitanio, A. (1999).
\newblock Statistical applications of the multivariate skew normal
  distribution.
\newblock {\em Journal of the Royal Statistical Society: Series B (Statistical
  Methodology)}, 61(3):579--602.

\bibitem[\protect\astroncite{Azzalini and
  Valle}{1996}]{azzalini1996multivariate}
Azzalini, A. and Valle, A.~D. (1996).
\newblock The multivariate skew-normal distribution.
\newblock {\em Biometrika}, 83(4):715--726.

\bibitem[\protect\astroncite{Benaroya et~al.}{2005}]{benaroya2005probability}
Benaroya, H., Han, S.~M., and Nagurka, M. (2005).
\newblock {\em Probability models in engineering and science}, volume 192.
\newblock CRC press.

\bibitem[\protect\astroncite{Benson}{1995}]{benson1995concave}
Benson, H.~P. (1995).
\newblock Concave minimization: theory, applications and algorithms.
\newblock In {\em Handbook of global optimization}, pages 43--148. Springer.

\bibitem[\protect\astroncite{Bernacchia and
  Naveau}{2008}]{bernacchia2008detecting1}
Bernacchia, A. and Naveau, P. (2008).
\newblock Detecting spatial patterns with the cumulant function--part 1: The
  theory.
\newblock {\em Nonlinear Processes in Geophysics}, 15(1):159--167.

\bibitem[\protect\astroncite{Bernacchia
  et~al.}{2008}]{bernacchia2008detecting2}
Bernacchia, A., Naveau, P., Vrac, M., and Yiou, P. (2008).
\newblock Detecting spatial patterns with the cumulant function--part 2: An
  application to el nino.
\newblock {\em Nonlinear Processes in Geophysics}, 15(1):169--177.

\bibitem[\protect\astroncite{Cesarone et~al.}{2020}]{cesarone2020stability}
Cesarone, F., Mango, F., Mottura, C.~D., Ricci, J.~M., and Tardella, F. (2020).
\newblock On the stability of portfolio selection models.
\newblock {\em Journal of Empirical Finance}, 59:210--234.

\bibitem[\protect\astroncite{Copson and Copson}{2004}]{copson2004asymptotic}
Copson, E.~T. and Copson, E.~T. (2004).
\newblock {\em Asymptotic expansions}.
\newblock Number~55. Cambridge university press.

\bibitem[\protect\astroncite{Das and Sinha}{1986}]{das1986detection}
Das, R. and Sinha, B.~K. (1986).
\newblock Detection of multivariate outliers with dispersion slippage in
  elliptically symmetric distributions.
\newblock {\em The Annals of Statistics}, pages 1619--1624.

\bibitem[\protect\astroncite{Domino}{2020}]{domino2020multivariate}
Domino, K. (2020).
\newblock Multivariate cumulants in outlier detection for financial data
  analysis.
\newblock {\em Physica A: Statistical Mechanics and its Applications}, page
  124995.

\bibitem[\protect\astroncite{Donoho}{1982}]{donoho1982breakdown}
Donoho, D.~L. (1982).
\newblock Breakdown properties of multivariate location estimators.
\newblock Technical report, Technical report, Harvard University, Boston. URL
  http://www-stat. stanford~….

\bibitem[\protect\astroncite{Ferguson}{1961}]{ferguson1961rejection}
Ferguson, T.~S. (1961).
\newblock On the rejection of outliers.
\newblock In {\em Proceedings of the fourth Berkeley symposium on mathematical
  statistics and probability}, volume~1, pages 253--287. University of
  California Press Berkeley and Los Angeles.

\bibitem[\protect\astroncite{Giacometti et~al.}{2021}]{giacometti2021tail}
Giacometti, R., Torri, G., and Paterlini, S. (2021).
\newblock Tail risks in large portfolio selection: penalized quantile and
  expectile minimum deviation models.
\newblock {\em Quantitative Finance}, 21(2):243--261.

\bibitem[\protect\astroncite{Gnanadesikan and
  Kettenring}{1972}]{gnanadesikan1972robust}
Gnanadesikan, R. and Kettenring, J.~R. (1972).
\newblock Robust estimates, residuals, and outlier detection with multiresponse
  data.
\newblock {\em Biometrics}, pages 81--124.

\bibitem[\protect\astroncite{Knuth}{2014}]{knuth2014art}
Knuth, D.~E. (2014).
\newblock {\em Art of computer programming, volume 2: Seminumerical
  algorithms}.
\newblock Addison-Wesley Professional.

\bibitem[\protect\astroncite{Kondor et~al.}{2007}]{Kondor2007}
Kondor, I., Pafka, S., and Nagy, G. (2007).
\newblock Noise sensitivity of portfolio selection under various risk measures.
\newblock {\em Journal of Banking \& Finance}, 31(5):1545--1573.

\bibitem[\protect\astroncite{Maronna}{1976}]{maronna1976robust}
Maronna, R.~A. (1976).
\newblock {Robust M-estimators of multivariate location and scatter}.
\newblock {\em The Annals of Statistics}, pages 51--67.

\bibitem[\protect\astroncite{Meng et~al.}{2019}]{meng2019overview}
Meng, F., Yuan, G., Lv, S., Wang, Z., and Xia, S. (2019).
\newblock An overview on trajectory outlier detection.
\newblock {\em Artificial Intelligence Review}, 52(4):2437--2456.

\bibitem[\protect\astroncite{Nocedal and Wright}{1999}]{nocedal1999numerical}
Nocedal, J. and Wright, S.~J. (1999).
\newblock {\em Numerical optimization}.
\newblock Springer.

\bibitem[\protect\astroncite{Pardalos and Rosen}{1986}]{pardalos1986methods}
Pardalos, P.~M. and Rosen, J.~B. (1986).
\newblock Methods for global concave minimization: A bibliographic survey.
\newblock {\em Siam Review}, 28(3):367--379.

\bibitem[\protect\astroncite{Pe{\~n}a and Prieto}{2001}]{pena2001multivariate}
Pe{\~n}a, D. and Prieto, F.~J. (2001).
\newblock Multivariate outlier detection and robust covariance matrix
  estimation.
\newblock {\em Technometrics}, 43(3):286--310.

\bibitem[\protect\astroncite{Pe{\~n}a and Prieto}{2007}]{pena2007combining}
Pe{\~n}a, D. and Prieto, F.~J. (2007).
\newblock Combining random and specific directions for outlier detection and
  robust estimation in high-dimensional multivariate data.
\newblock {\em Journal of Computational and Graphical Statistics},
  16(1):228--254.

\bibitem[\protect\astroncite{Reed and Yu}{1990}]{reed1990adaptive}
Reed, I.~S. and Yu, X. (1990).
\newblock Adaptive multiple-band cfar detection of an optical pattern with
  unknown spectral distribution.
\newblock {\em IEEE transactions on acoustics, speech, and signal processing},
  38(10):1760--1770.

\bibitem[\protect\astroncite{Rockafellar}{1970}]{rockafellar1970convex}
Rockafellar, R.~T. (1970).
\newblock {\em Convex Analysis}, volume~36.
\newblock Princeton University Press.

\bibitem[\protect\astroncite{Rousseeuw}{1984}]{rousseeuw1984least}
Rousseeuw, P.~J. (1984).
\newblock Least median of squares regression.
\newblock {\em Journal of the American Statistical Association},
  79(388):871--880.

\bibitem[\protect\astroncite{Rousseeuw}{1985}]{rousseeuw1985multivariate}
Rousseeuw, P.~J. (1985).
\newblock Multivariate estimation with high breakdown point.
\newblock {\em Mathematical Statistics and Applications}, 8(283-297):37.

\bibitem[\protect\astroncite{Schwager and
  Margolin}{1982}]{schwager1982detection}
Schwager, S.~J. and Margolin, B.~H. (1982).
\newblock Detection of multivariate normal outliers.
\newblock {\em The Annals of Statistics}, 10(3):943--954.

\bibitem[\protect\astroncite{Shu et~al.}{2021}]{shu2021covid}
Shu, M., Song, R., and Zhu, W. (2021).
\newblock {The ``COVID'' crash of the 2020 US Stock market}.
\newblock {\em The North American Journal of Economics and Finance}, 58:101497.

\bibitem[\protect\astroncite{Singh and Upadhyaya}{2012}]{singh2012outlier}
Singh, K. and Upadhyaya, S. (2012).
\newblock Outlier detection: applications and techniques.
\newblock {\em International Journal of Computer Science Issues (IJCSI)},
  9(1):307.

\bibitem[\protect\astroncite{Sinha}{1984}]{sinha1984detection}
Sinha, B.~K. (1984).
\newblock Detection of multivariate outliers in elliptically symmetric
  distributions.
\newblock {\em The Annals of Statistics}, pages 1558--1565.

\bibitem[\protect\astroncite{Stahel}{1981}]{stahel1981breakdown}
Stahel, W.~A. (1981).
\newblock {\em Breakdown of covariance estimators}.
\newblock Fachgruppe f{\"u}r Statistik, Eidgen{\"o}ssische Techn. Hochsch.

\bibitem[\protect\astroncite{Trendafilov and
  Jolliffe}{2006}]{trendafilov2006projected}
Trendafilov, N.~T. and Jolliffe, I.~T. (2006).
\newblock Projected gradient approach to the numerical solution of the
  scotlass.
\newblock {\em Computational Statistics \& Data Analysis}, 50(1):242--253.

\bibitem[\protect\astroncite{Wilks}{1963}]{wilks1963multivariate}
Wilks, S.~S. (1963).
\newblock Multivariate statistical outliers.
\newblock {\em Sankhy{\=a}: The Indian Journal of Statistics, Series A}, pages
  407--426.

\end{thebibliography}

}

\end{document}